%% file: paper.tex
\documentclass{sig-alternate-tr}
\usepackage[usenames,dvipsnames]{color}
\definecolor{mycitegreen}{RGB}{0,90,0}
\definecolor{myrefblue}{RGB}{90,0,0}
\definecolor{cyred}{RGB}{140,0,0}
\definecolor{commentgreen}{RGB}{0,80,0}
\usepackage[draft=true, final=true, colorlinks=true, citecolor=mycitegreen, linkcolor=myrefblue, anchorcolor=myrefblue, pagebackref=false]{hyperref}
\usepackage{graphics,enumitem}
\usepackage{epsfig}
\usepackage[captionskip=0pt, font={small}]{subfig}
\usepackage{color}
\usepackage{url}
\usepackage{xspace}
\usepackage{flushend}
\def\compactify{\itemsep=0pt \topsep=0pt \partopsep=0pt \parsep=0pt}
 \let\latexusecounter=\usecounter

\makeatletter
\newif\if@restonecol
\makeatother

\usepackage[lined,vlined,boxruled]{algorithm2e}

\title{Finishing Flows Quickly with Preemptive Scheduling\thanks{The conference version was published in \emph{SIGCOMM'12}.}}

\numberofauthors{3}
\author{
\alignauthor
Chi-Yao Hong\\
\affaddr{UIUC}\\
\email{cyhong@illinois.edu}
\alignauthor
Matthew Caesar\\
\affaddr{UIUC}\\
\email{caesar@illinois.edu}
\alignauthor
P. Brighten Godfrey\\
\affaddr{UIUC}\\
\email{pbg@illinois.edu}
}

\begin{document}

%\conferenceinfo{SIGCOMM'12,} {August 13--17, 2012, Helsinki, Finland.} 
%\CopyrightYear{2012} 
%\crdata{978-1-4503-1419-0/12/08} 
%\clubpenalty=10000 
%\widowpenalty = 10000

\newcounter{copyrightbox}

\maketitle
%\baselineskip=12bp

% Paragraph Highlight
\newcommand{\paragraphb}[1]{\vspace{0.03in}\noindent{\bf #1} }
\newcommand{\paragraphe}[1]{\vspace{0.03in}\noindent{\em #1} }
\newcommand{\paragraphbe}[1]{\vspace{0.03in}\noindent{\bf \em #1} }

% Comment colors
%\newcommand{\comment}{\textcolor{black}}
\newcommand{\prooffontsize}{\fontsize{8pt}{9.3pt}\selectfont}
\newcommand{\cready}{\textcolor{black}}
\newcommand{\newcomment}{\textcolor{black}}
\newcommand{\cycomment}{\textcolor{black}}
\newcommand{\matt}[1]{{\color{black}{#1}}}
\newcommand{\codecomment}{\textcolor{commentgreen}}
\newcommand{\sys}{PDQ\xspace}
\newcommand{\mpsys}{M-PDQ\xspace}
\newcommand{\cut}[1]{}
\newcommand{\figtitle}[1]{\small\emph{\textbf{#1}}}
\newcommand{\bntl}{D$^3$\xspace}

\newcommand{\objtolbl}{\vspace{-12pt}}
\newcommand{\objtoobj}{\vspace{-12.5pt}}
\newcommand{\algtotext}{\vspace{-8pt}}
\newcommand{\algfontsize}{\fontsize{9.5pt}{10.5pt}\selectfont}
\newcommand{\listskip}{\vspace{-5pt}}
\newcommand{\listskips}{\vspace{-3pt}}
\newcommand{\secendskip}{\vspace{-3pt}}

% To make the FIXMEs go away, comment out this line...
\newcommand{\fixme}[1]{{\bf\textcolor{red}{[#1]}}}
% ...and uncomment this one.
%\newcommand{\fixme}[1]{}

\newcommand{\helpme}[1]{{\bf\textcolor{red}{#1}}}

\input{abstract}

\input{intro}

\input{overview}

\input{algorithm}
\input{properties}
\input{evaluation}

\input{badcases}
%\input{dynamic}
\input{multipath}

%\input{implementation}
%\input{deployment}
\input{discussion}

\input{related}

\input{conclusion}
{
%\scriptsize
\vspace{-5pt}
\bibliography{paper}
}
\bibliographystyle{abbrv}
%\bibliographystyle{plain}

\input{appendix}

\end{document}

%% file: abstract.tex
\begin{abstract}
\textnormal{Today's data centers face extreme challenges in providing low latency.  
However, fair sharing, a principle commonly adopted in current congestion control protocols, is far 
from optimal for satisfying latency requirements.}

\textnormal{
We propose Preemptive Distributed Quick (\textbf{\sys{}}) flow scheduling, a protocol designed to complete flows quickly and meet flow deadlines.
\sys enables flow preemption to approximate a range of scheduling disciplines.
For example, \sys can emulate a shortest job first algorithm to give priority to the short flows by pausing the contending flows. 
\sys borrows ideas from centralized scheduling disciplines and implements them in a fully distributed manner, making it scalable to 
today's data centers. Further, we develop a multipath version of \sys to exploit path diversity.}

\textnormal{Through extensive packet-level and flow-level simulation, we demonstrate that \sys significantly 
outperforms TCP, RCP and D$^3$ in data center environments. 
We further show that \sys is stable, resilient to packet loss, and preserves nearly all its performance gains even given inaccurate flow information.}
\end{abstract}

%A category including the fourth, optional field follows...
%\category{C.2.2}{Computer-Communication Networks}[Network Protocols]
%\terms{Performance}

\paragraphb{Categories and Subject Descriptors:} C.2.2 [Computer-Communication Networks]: Network Protocols

\vspace{-2pt}
\paragraphb{General Terms:} Algorithms, Design, Performance

\vspace{-10pt}

%% file: intro.tex
\section{Introduction}
\label{sec:introduction}

Data centers are now used as the underlying infrastructure of many modern commercial operations, including web services, cloud computing, and some of the world's largest databases and storage services. Data center applications including financial services, social networking, recommendation systems, and web search often have very demanding latency requirements.  For example, even fractions of a second make a quantifiable difference in user experience for web services~\cite{brutlag09}.  And a service that aggregates results from many back-end servers has even more stringent requirements on completion time of the back-end flows, since the service must often wait for the {\em last} of these flows to finish or else reduce the quality of the final results.\footnote{See discussion in \cite{dctcp}, \S2.1.}  Minimizing delays from network congestion, or meeting soft-real-time deadlines with high probability, is therefore important.

Unfortunately, current transport protocols neither minimize flow completion time nor meet deadlines. TCP, RCP~\cite{rcp}, ICTCP~\cite{ictcp}, and DCTCP~\cite{dctcp} approximate \emph{fair sharing}, dividing link bandwidth equally among flows. Fair sharing is known to be far from optimal in terms of minimizing flow completion time~\cite{PShaslargeFCT} and the number of deadline-missing flows~\cite{EDFbetterthanPS}. \cready{As a result, a study of three production data centers~\cite{bntl} showed that a significant fraction ($7 - 25\%$) of flow deadlines were missed}, resulting in degradation of application response quality, waste of network bandwidth, and ultimately loss of operator revenue~\cite{dctcp}.

This paper introduces \textbf{Preemptive Distributed Quick (\sys)} flow scheduling, a protocol designed to complete flows quickly and meet flow deadlines.  \sys builds on traditional real-time scheduling techniques: when processing a queue of tasks, scheduling in order of Earliest Deadline First (EDF) is known to minimize the number of late tasks, while Shortest Job First (SJF) minimizes mean flow completion time.  However, applying these policies to scheduling data center flows introduces several new challenges.

First, EDF and SJF assume a centralized scheduler which knows the global state of the system; this would impede our goal of low latency in a large data center. To perform dynamic decentralized scheduling, \sys provides a distributed algorithm to allow a set of switches to collaboratively gather information about flow workloads and converge to a stable agreement on allocation decisions. Second, unlike ``fair sharing'' protocols, EDF and SJF rely on the ability to {\em preempt} existing tasks, to ensure a newly arriving task with a smaller deadline can be completed before a currently-scheduled task. To support this functionality in distributed environments, \sys provides the ability to perform distributed preemption of existing flow traffic, in a manner that enables fast switchover and is guaranteed to never deadlock.

Thus, \sys provides a distributed flow scheduling layer which is {\em lightweight}, using only FIFO tail-drop queues, and {\em flexible},  in that it can approximate a range of scheduling disciplines based on relative priority of flows.  We use this primitive to implement two scheduling disciplines: EDF to minimize mean flow completion time, and SJF to minimize the number of deadline-missing flows.

Through an extensive simulation study using real datacenter workloads, we find that \sys provides strong benefits over existing datacenter transport mechanisms. \sys is most closely related to D$^3$~\cite{bntl}, which also tries to meet flow deadlines. Unlike D$^3$, which is a ``first-come first-reserve'' algorithm, \sys proactively and preemptively gives network resources to the most critical flows.  For deadline-constrained flows, our evaluation shows \sys supports $3$ times more concurrent senders than \cite{bntl} while satisfying their flow deadlines. When flows have no deadlines, we show \sys can reduce mean flow completion times by $\sim$$30\%$ or more compared with TCP, RCP, and D$^3$.

The key contributions of this paper are:
\begin{itemize}
\listskip\item We design and implement \sys, a distributed flow scheduling layer for data centers which can approximate a range of scheduling disciplines.
\listskip\item We build on \sys to implement flow scheduling disciplines that minimize mean flow completion time and the number of deadline-missing flows.
\listskip\item We demonstrate \sys can save $\sim$$30\%$ average flow completion time compared with TCP, RCP and \bntl; and can support $3\times$ as many concurrent senders as \bntl while meeting flow deadlines.
\listskip\item We show that \sys is stable, resilient to packet loss, and preserves nearly all its performance gains even given inaccurate flow information.
\listskip\item We develop and evaluate a multipath version of \sys, showing further performance and reliability gains.
\end{itemize}

%% file: overview.tex
\section{Overview}
\label{sec:overview}

We start by presenting an example to demonstrate potential benefits of \sys
over existing approaches~(\S\ref{sec:overview:benefits}). 
We then give a description of key challenges that \sys must address~(\S\ref{sec:overview:challenges}).

%We present an example of benefits to demonstrate potential performance improvements of \sys over existing 
%methods~(\S\ref{sec:overview:benefits}), followed by a description of key challenges \sys faces~(\S\ref{sec:overview:challenges}).

\subsection{Example of Benefits}
\label{sec:overview:benefits}

Consider the scenario shown in Figure~\ref{fig:ex}, where three concurrent
flows ($f_A$, $f_B$, and $f_C$) arrive \matt{simultaneously.}
%at the same time.

\paragraphb{Deadline-unconstrained Case:} Suppose that the flows have no deadlines, and our
objective is to minimize the average flow completion time.
Assuming a fluid traffic model (infinitesimal units of transmission),
the result given by fair sharing is shown in Figure~\ref{fig:ex1}:
[$f_A$,$f_B$,$f_C$] finish at time [$3$,$5$,$6$], and the average flow completion time is $\frac{3+5+6}{3}=4.67$.
If we schedule the flows by SJF (one by one according to flow size), as shown in Figure~\ref{fig:ex2},
the completion time becomes $\frac{1+3+6}{3}=3.33$, a savings of $\sim$$29\%$ compared to fair sharing.
Moreover, for every individual flow, the flow completion time in SJF is no larger than that given by fair sharing.

\begin{figure}[t]
\centering
\hspace{-5pt}
\subfloat[]{\small\label{fig:ex:parameter}
\begin{tabular}[b]{c|c|c}
Flow ID & Size & Deadline \\\hline\hline
$f_A$   & $1$  &  $1$ \\
$f_B$   & $2$  &  $4$ \\
$f_C$   & $3$  &  $6$
\end{tabular}
}
\subfloat[]{ \label{fig:ex1} \includegraphics[width=1.6in]{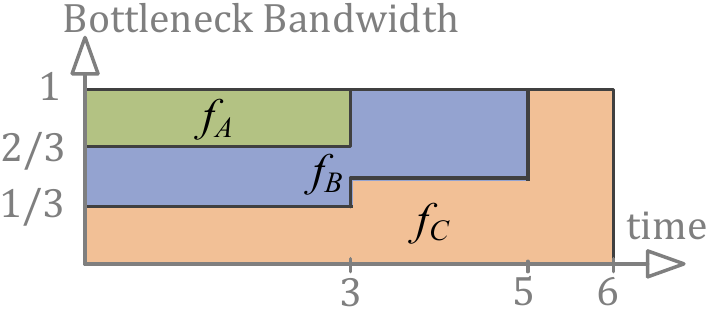}}\objtoobj\\
\subfloat[]{ \label{fig:ex2} \includegraphics[width=1.6in]{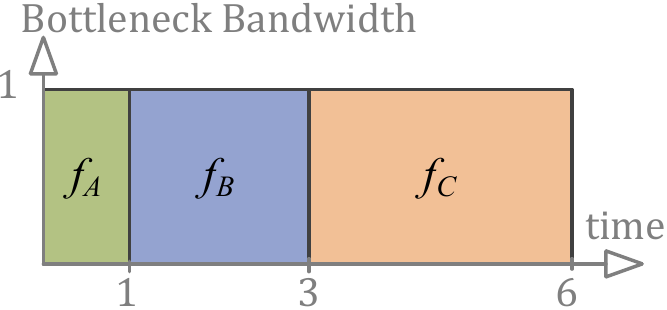}}\hspace{-5pt}
\subfloat[]{ \label{fig:ex3} \includegraphics[width=1.6in]{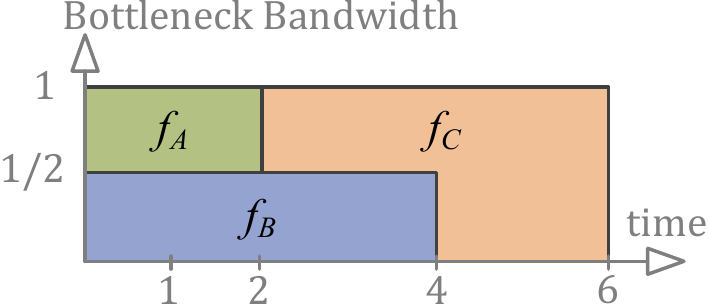}}\objtolbl
\caption{\figtitle{Motivating Example. (a) Three concurrent flows competing for a single bottleneck link; (b) Fair sharing; (c) SJF/EDF;
(d) \bntl for flow arrival order $f_B$$\leadsto$$f_A$$\leadsto$$f_C$. }}
\label{fig:ex}
\end{figure}

\paragraphb{Deadline-constrained Case:} Suppose now that the flows have deadlines, as specified in Figure~\ref{fig:ex:parameter}.
The objective becomes \matt{minimizing} the number of tardy flows, i.e., \matt{maximizing} the number of flows that meet their deadlines.
For fair sharing, both flow $f_A$ and $f_B$ fail to meet their deadlines, as shown in Figure~\ref{fig:ex1}.
If we schedule the flows by EDF (one by one according to deadline), as shown in Figure~\ref{fig:ex2},
every flow can finish before its deadline.

Now we consider \bntl, a recently proposed deadline-aware protocol for data center networks~\cite{bntl}.
When the network is congested, \bntl satisfies as many flows as
possible according to the flow request rate in the order of their arrival. In particular, each flow will request 
a rate $r=\frac{s}{d}$, where $s$ is the flow's size and $d$ is the time until its deadline.
Therefore, the result of \bntl depends highly on flow arrival order.
Assuming flows arrive in the order $f_B$$\leadsto$$f_A$$\leadsto$$f_C$, the result of \bntl is shown in Figure~\ref{fig:ex3}.
Flow $f_B$ will send with rate $\frac{2}{4}=0.5$ and will finish right before its deadline.
However, flow $f_A$, which arrives later than $f_B$, will fail to meet its deadline using the
remaining bandwidth, as evident in Figure~\ref{fig:ex3}.
In fact, out of $3!=6$ possible permutations of flow arrival order, \bntl will fail to satisfy some of the deadlines for $5$ cases, the only exception
being the order $f_A$$\leadsto$$f_B$$\leadsto$$f_C$, which is the order EDF chooses.
Although \bntl also allows senders to terminate flows that fail to meet their deadlines (to save bandwidth), termination does not help in this scenario and is not
presented in Figure~\ref{fig:ex3}.

%n D3 , flow fA might occupy the
%remaining bandwidth and terminates (with partial un-
%delivered data) at time ∼0.5 when the sender realizes it
%cannot meet deadline anyway.

\subsection{Design Challenges}
\label{sec:overview:challenges}
Although attractive performance gains are seen from the example, many design challenges remain to realize the expected benefits.

\paragraphb{Decentralizing Scheduling Disciplines:}
Scheduling disciplines like EDF or SJF are centralized algorithms that require
global knowledge of flow information, introducing a single point of failure and significant overhead
for senders to interact with the centralized coordinator. For example, a centralized scheduler introduces
considerable flow initialization overhead, while becoming
a congestive hot-spot. This problem is especially severe in data center workloads where the majority of flows are mice.
A scheduler maintaining only elephant flows like DevoFlow~\cite{devoflow} seems unlikely to succeed in congestion control as deadline
constraints are usually associated with mice.
%Moreover, many existing scheduling algorithms (e.g., \cite{Palencia:05}) do not solve our problem
%as they use \emph{centralized} algorithms that schedule jobs on a distributed system.
%\mattc{This sentence leads to the question: why can't you do the same thing? 
%Schemes like \cite{Palencia:05} solve the problems you mention above by 
%using control messages to bring state to the central coordinator, in a manner
%that is fault tolerant and scalable enough. What's wrong with the Palencia approach?
%I think this can be addressed quite simply by just adding a few words or moving
%this sentence.}
The need to address the above limitations leads to \sys, a fully distributed solution where switches collaboratively control flow schedules.

\paragraphb{Switching Between Flows Seamlessly:}
The example of \S\ref{sec:overview:benefits} idealistically assumed
we can start a new flow immediately after a previous one terminates, enabling
all the transmission schedules (Figure~\ref{fig:ex1}, \ref{fig:ex2} and \ref{fig:ex3}) to
fully utilize the bottleneck bandwidth and thus complete flows as quickly as possible. 
%, i.e., the total throughput is maximized.
However, achieving high utilization during flow switching in practice requires precise control of flow
transmission time.
One could simplify this problem by assuming synchronized time among
both switches and senders, but that introduces additional cost and effort to coordinate clocks.
\sys addresses this problem by \cycomment{starting the next 
set of waiting flows slightly before the current sending flows finish.}

% Chi-Yao: "slightly before when something is about to finish" --> too vague

%are about to
%finish.  % matt: or, "\sys addresses this problem through {\em seamless flow switching}: the next set of waiting flows are started slightly..."
%To provide seamless flow switching, \sys tackles this problem by starting the next set of waiting flows slightly before when the current sending flows are about to finish.

%However, high utilization does not ensure low flow completion time, as the fair sharing case shown in the example (Figure~\ref{fig:ex2}).
%\mattc{This sentence says a problem, but then the next paragraph seems to change
%the subject. Do you ever address the problem mentioned in this sentence?
%I think you should either delete this last sentence or add a few words to
%connect it with the next paragraph (eg by mentioning explicitly why high
%utilization isn't enough).}

\paragraphb{Prioritizing Flows using FIFO Tail-drop Queues:}
One could implement priority queues in switches to approximate flow scheduling by enforcing packet priority.
Ideally, this requires that each of the concurrent flows has a unique priority class. 
\cycomment{However, a data center switch can have several thousand active flows within a one second bin~\cite{imcDCmeasure}, while modern 
switches support only $\sim$$10$ priority classes~\cite{bntl}}.
Therefore, for today's data center switches, the number of priority classes per port is far below the requirements of such an approach, and it is unclear 
whether modifying switches to support a larger number of priority classes can be cost-effective. 
To solve this, \sys explicitly controls the flow sending rate to regulate flow traffic and retain packets from low-priority flows at senders.
With this flow pausing strategy, \sys only requires per-link FIFO tail-drop queues at switches.

\secendskip

%% file: algorithm.tex
\newcommand{\maxrate}{$\mathcal{R}$\xspace}
\newcommand{\maxratez}{$\mathcal{R}^{\max}$\xspace}
\newcommand{\maxrateh}{\maxrate{}$_{H}^{\max}$\xspace}
\newcommand{\maxrates}{\maxrate{}$_{S}^{\max}$\xspace}
\newcommand{\maxratei}{\maxrate{}$_{i}^{\max}$\xspace}
\newcommand{\maxratej}{\maxrate{}$_{j}^{\max}$\xspace}

\newcommand{\schrate}{$\mathcal{R}$\xspace}
\newcommand{\schrateh}{\maxrate{}$_{H}^{\textrm{sch}}$\xspace}
\newcommand{\schrates}{\maxrate{}$_{S}^{\textrm{sch}}$\xspace}
\newcommand{\schratei}{\maxrate{}$_{i}^{\textrm{sch}}$\xspace}
\newcommand{\schratej}{\maxrate{}$_{j}^{\textrm{sch}}$\xspace}

\newcommand{\rtt}{$RTT$\xspace}
\newcommand{\rtth}{\rtt{}$_{H}$\xspace}
\newcommand{\rtts}{\rtt{}$_{S}$\xspace}
\newcommand{\rtti}{\rtt{}$_{i}$\xspace}
\newcommand{\rttj}{\rtt{}$_{j}$\xspace}

\newcommand{\rate}{$\mathcal{R}$\xspace}
\newcommand{\rateh}{\rate{}$_{H}$\xspace}
\newcommand{\rates}{\rate{}$_{S}$\xspace}
\newcommand{\ratei}{\rate{}$_{i}$\xspace}
\newcommand{\ratej}{\rate{}$_{j}$\xspace}
\newcommand{\pauseby}{$\mathcal{P}$\xspace}
\newcommand{\pausebyh}{\pauseby{}$_{H}$\xspace}
\newcommand{\pausebys}{\pauseby{}$_{S}$\xspace}
\newcommand{\pausebyi}{\pauseby{}$_{i}$\xspace}
\newcommand{\deadline}{$\mathcal{D}$\xspace}
\newcommand{\deadlineh}{\deadline{}$_{H}$\xspace}
\newcommand{\deadlines}{\deadline{}$_{S}$\xspace}
\newcommand{\deadlinei}{\deadline{}$_{i}$\xspace}
\newcommand{\size}{$\mathcal{S}$\xspace}
\newcommand{\sizeh}{\size{}$_{H}$\xspace}
\newcommand{\sizes}{\size{}$_{S}$\xspace}
\newcommand{\sizei}{\size{}$_{i}$\xspace}
\newcommand{\retime}{$\mathcal{T}$\xspace}
\newcommand{\retimeh}{\retime{}$_{H}$\xspace}
\newcommand{\retimes}{\retime{}$_{S}$\xspace}
\newcommand{\retimei}{\retime{}$_{i}$\xspace}
\newcommand{\pf}{$\mathcal{I}$\xspace}
\newcommand{\pfh}{\pf{}$_{H}$\xspace}
\newcommand{\pfs}{\pf{}$_{S}$\xspace}
\newcommand{\pfi}{\pf{}$_{i}$\xspace}

\newcommand{\snd}{$Active$\xspace}
\newcommand{\sndi}{\snd{}$_{i}$\xspace}

\section{Protocol}
\label{sec:algorithm}

\matt{
We first present an overview of our design. We then describe the design
details implemented at the sender (\S\ref{sec:algorithm:sender}), receiver (\S\ref{sec:algorithm:receiver}) and
switches (\S\ref{sec:algorithm:switch}). 
This section assumes each flow uses a single path.
In \S\ref{sec:multipath}, we will show how \sys can be extended to support multipath forwarding.
}

%We first present an overview of our design~(\S\ref{sec:algorithm:overview}), followed by a 
%description of design details implemented at the sender (\S\ref{sec:algorithm:sender}), receiver (\S\ref{sec:algorithm:receiver}) and 
%switches (\S\ref{sec:algorithm:switch}). This section assumes each flow uses a single path.
%In \S\ref{sec:multipath}, we will show how \sys can be extended to support multipath forwarding.

%For clarity, we first present a centralized version of 
%our design~(\S\ref{sec:algorithm:cent}).
%%We first present a centralized proposal~(\S\ref{sec:algorithm:cent}).
%\footnote{This assumption also implies that hosts have perfect synchronized clocks.}
%Then we describe our fully distributed algorithm, including a protocol overview (\S\ref{sec:algorithm:overview}), followed by a joint design of 
%end servers (\S\ref{sec:algorithm:sender} and \S\ref{sec:algorithm:receiver}) and 
%switches (\S\ref{sec:algorithm:switch}). This 
%section assumes each flow uses a single path. 

%\subsection{Overview}
%\label{sec:algorithm:overview}

\paragraphb{Centralized Algorithm:} To clarify \matt{our approach,} we start by presenting it as an idealized {\em centralized} scheduler with
complete visibility of the network, able to communicate with devices in
the network with {\em zero delay}. To simplify exposition, the centralized scheduler assumes that
flows have no deadlines, and our only goal is to optimize flow completion time. 
We will later relax these assumptions.

%%We first present an algorithm which implements \sys under certain assumptions. In particular,
%%this algorithm assumes a \emph{centralized} scheduler with \emph{no feedback delay} in the congestion control loop.
%%To simplify the problem, we also assume that flows have no deadline in this proposal. 
%%We will relax them and propose a fully distributed design in the next sections.

We define the \emph{expected flow transmission time}, denoted by \retimei for any flow $i$, to be the remaining 
flow size divided by its \emph{maximal sending rate} \maxratei.
The maximal sending rate \maxratei is the minimum of the sender NIC rate, the switch link 
rates, and the rate that receiver can process and receive. 
Whenever network workload changes (a new flow arrives, or an existing flow terminates), the centralized scheduler recomputes the
flow transmission schedule as follows: 
\begin{enumerate}
\listskip\item $B_e =$ available bandwidth of link $e$, initialized to $e$'s line rate.	
\listskip\item For each flow $i$, in increasing order of \retimei:	
		\begin{enumerate}[topsep=0pt, partopsep=0pt]
			\item \listskips Let $P_i$ be flow $i$'s path.
			\item \listskips Send flow $i$ with rate \schratei$ = \min_{\forall e \in P_i}(\textrm{\maxratei}, B_e)$. 
			\item \listskips $B_e \gets B_e - $ \schratei for each $e \in P_i$.
		\end{enumerate}
\listskip
\end{enumerate}

\cut{\begin{description}
  \item[Step 1.] Sort the flows in non-decreasing order of \retime. Let $B_e$ be the 
available bandwidth of any link $e$, initialized to the link line rate. 
Perform Step 2 for each flow $i$.
  \item[Step 2.] Send flow $i$ with rate \schratei$ = \min_{\forall e}(\textrm{\maxratei}, B_e)$, where $e$ runs 
over every link along the routing path. %We refer to this rate as the \emph{scheduling rate}. 
Subtract \schratei from each $B_e$ that flow $i$ traverses.
\end{description}}

%However, this algorithm makes several unrealistic assumptions. 
%The problems with these assumptions are discussed in~\S\ref{sec:challenges}. 
%\subsection{\sys Overview}
%\label{sec:algorithm:overview}
%First we briefly overview how \sys realizes scheduling disciplines in distributed way.

\paragraphb{Distributed Algorithm:} We eliminate the unrealistic assumptions we made in the centralized algorithm to construct a fully
distributed realization of our design. To distribute its operation,
\sys switches propagate flow information via explicit feedback in packet headers.
\sys senders maintain a set of flow-related variables such as flow sending rate and flow size and communicate the flow information 
to the intermediate switches via a scheduling header added to the transport layer of each data packet.
%\cycomment{added to every data packet in the transport layer}. 
When the feedback reaches the receiver, it is returned to the sender in an ACK packet.
\sys switches monitor the incoming traffic rate of each of their output queues and inform the sender to send data with 
a specific rate (\rate{}$>$$0$) or to pause (\rate{}$=$$0$) by annotating the scheduling header of data/ACK packets.
We present the details of this distributed realization in the following sections.

\subsection{\sys Sender}
\label{sec:algorithm:sender}

Like many transport protocols, a \sys sender sends a SYN packet for 
flow initialization and a TERM packet for flow termination, and resends a packet after a timeout. 
The sender maintains standard data structures for reliable 
transmission, including estimated round-trip time and states (e.g., timer) for in-flight packets.
The \sys sender maintains several state variables: its current sending 
rate (\rates, initialized to zero), the ID of the switch (if any) who has paused the flow (\pausebys, initialized to \o), 
flow deadline (\deadlines, which is optional), the expected flow transmission time (\retimes, initialized to the flow size divided by sender NIC rate), 
%the remaining flow size (\sizes), expected remaining processing time (), 
the inter-probing time (\pfs, initialized to \o), and the measured RTT (\rtts, estimated by an exponential decay).

The sender sends packets with rate \rates. 
%If the rate is zero, switches are unable to inform the end host of the proper sending rate by annotating packets, and thus the 
%the sender sends a \emph{probe} packe teven when \rates $ = 0$ -- Chi-Yao: it's not "even when" because senders don't send probes when rates > 0.
If the rate is zero, the sender sends a \emph{probe} packet every $\textrm{\pfs}$ RTTs to get rate information from the switches.
A probe packet is a packet with a scheduling header but no data content.

On packet departure, the sender attaches a scheduling header to 
the packet, containing fields set based on the values of each of the sender's state variables above.
%variables corresponding to each of the sender's state variables above. 
\rateh is always set to the maximal sending rate \maxrates, while the remaining fields in the scheduling header are set 
to its current maintained variables.
%their flow sending rate (\rate), the ID of the switch who pauses the 
%flow (\pauseby), the flow deadline (\deadline), the remaining flow size (\size), the expected remaining processing time (\retime), and 
%the
Note that the subscript $H$ refers to a field in the scheduling \underline{h}eader; 
the subscript $S$ refers to a variable maintained by the \underline{s}ender; 
the subscript $i$ refers to a variable related to the $i$th flow in the switch's flow list.

Whenever an ACK packet arrives, the sender updates its flow sending rate based on the 
feedback: \retimes is updated based on the remaining flow size, \rtts is updated based on the packet arrival time, \matt{and}
the remaining variables are set to the fields in the scheduling header.

%\rates is set to \rateh; 
%\pausebys is set to \pausebyh; 
%\retimes is updated based on the remaining flow size.
%\pfs is set to \pfh; 
%\sizes is subtracted by the ACKed data size; 
%transmission time in terms of the estimated round-trip time: $\frac{\textrm{\sizes}}{\textrm{NIC\_Rate}}/\textrm{RTT}$.

\paragraphb{Early Termination:} For deadline-constrained flows, when 
the incoming flow demand exceeds the network capacity, there might not exist a feasible schedule 
for all flows to meet their deadlines. In this case, it is desirable to discard a minimal number of flows 
while satisfying the deadline of the remaining flows. Unfortunately, minimizing the number of 
tardy flows in a dynamic setting is an $\mathcal{NP}$-complete problem.\footnote{Consider a 
subproblem where a set of concurrent flows that share a bottleneck link all have the same deadline. 
This subproblem of minimizing the number of tardy flows is exactly the $\mathcal{NP}$-complete subset sum problem~\cite{NPCproblem}.
}

Therefore, we use a simple heuristic, called \emph{Early Termination}, to terminate a flow when it cannot meet its deadline. \matt{Here,} %In particular, 
the sender sends a TERM packet whenever \emph{any} of the following conditions \matt{happen}:
\begin{enumerate}
\item \listskip Deadline is past (Time $>$ \deadlines).
\item \listskip The remaining flow transmission time is larger than the time to deadline (Time $+$ \retimes $>$ \deadlines).
%\sizes $>$ $($\deadlines $-$ Time$)$ $\times$ NIC\_Rate$)$.
\item \listskip The flow is paused (\rates $=0$), and the time to deadline is smaller than an RTT (Time $+$ \rtts $>$ \deadlines).
\end{enumerate}

\subsection{\sys Receiver}
\label{sec:algorithm:receiver}
A \sys receiver copies the scheduling header from each data packet to its corresponding ACK. Moreover, to avoid 
the sender overrunning the receiver's buffer, the \sys receiver 
reduces \rateh if it exceeds the maximal rate that receiver can process and receive.

\subsection{\sys Switch}
\label{sec:algorithm:switch}

\cycomment{The high-level objective of a \sys switch is to let the most \emph{critical} flow complete as soon as possible. 
To this end, switches share a common flow comparator, which decides flow criticality, to approximate a range of scheduling disciplines.
In this study, we implement two disciplines, EDF and SJF, while we give higher priority to EDF.
In particular, we say a flow is more critical than another one if it has smaller deadline (emulating EDF to minimize the number of deadline-missing flows). 
When there is a tie or flows have no deadline, we break it by giving priority to the flow with smaller expected transmission 
time (emulating SJF to minimize mean flow completion time). If a tie remains, we break it by flow ID. 
If desired, the operator could easily override the comparator to approximate other scheduling 
disciplines. For example, we also evaluate another scheduling discipline incorporating flow waiting time in \S\ref{sec:discussion}.}

%If desired, one could easily override this {\em comparator}, which decides flow criticality, to emulate 
%other scheduling disciplines. 
%The only restriction we impose is that switches should share a common flow comparator.  \fixme{This is well written.  But if we are emphasizing generality in the intro, you might want to flip this around, and talk about the general comparator first, and then the special cases (that we implement)}

\cycomment{The switch's purpose is to resolve flow contention: flows can preempt less critical flows to achieve the highest possible sending rate. 
To achieve this goal, the switches maintain state about flows on 
each link (\S\ref{sec:algorithm:switch:s}) and exchange information by tagging the scheduling header.
To compute the rate feedback (\rateh), the switch uses a flow controller (controlling which flows to send; \S\ref{sec:algorithm:switch:fc}) and a rate 
controller (computing the aggregate flow sending rate; \S\ref{sec:algorithm:switch:rc}). }

\subsubsection{Switch State}
\label{sec:algorithm:switch:s}

In order to resolve flow contention, the switch maintains state about flows on each link.  Specifically, it remembers the most recent variables (<\ratei, \pausebyi, \deadlinei, \retimei, \rtti{}>) obtained from observed packet headers for flow $i$, which it uses to decide at any moment the correct sending rate for the flows.  However, we do not have to keep this state for {\em all} flows.  Specifically, \sys switches only store the most critical $2\kappa$ flows, where $\kappa$ is the 
number of \emph{sending} flows (i.e., flows with sending rate \rates$>0$). Since \sys allocates as much link bandwidth as possible to the most critical flows 
until the link is fully utilized, the $\kappa$ most critical flows fully utilize the link's bandwidth; we store state for $2\kappa$ flows in order to have sufficient information immediately available to un-pause another flow if one of the sending flows completes.  The remaining flows are not remembered by the switch, until they become sufficiently critical.
% matt: I guess the assumption here is all $2\kappa$ would most likely not all finish at the same time? REVISIT
% CY: There's usually only $\kappa$ flows are sending at a time, so it's unlikely all $2\kappa$ flows will finish at the same time.
% If it happens, it's just a waste of 1 - 2 RTTs (< 150 us in many DCN) until the next set of flows can start sending.

The amount of state maintained at the switch thus depends on how many flows are needed to fill up a link.  
In most practical cases, this value will be very small because \cycomment{(i) \sys allows critical flows to send with their highest possible rates, and (ii)} switch-to-switch links are typically only $1 - 10\times$ faster than server-to-switch 
links, e.g., current data center networks mostly use $1$ Gbps server links and $10$ Gbps switch links\footnote{For example, \matt{the} NEC PF5240 switch supports 48 $\times$ $1$ Gbps ports, along with $2\times10$ Gbps ports; Pronto 3290 switch provides $48\times 1$ Gbps ports and $4 \times 10$ Gbps ports.}, and the next 
generation will likely be $10$ Gbps server links and $40$ or $100$ Gbps switch links.  
However, if a flow's rate is limited to something less than its NIC rate (e.g., due to processing or disk bottlenecks), switches may \cycomment{need to} store more flows.

%<<<<<<< .mine
%\helpme{Today's switches are typically equipped with $4 - 16$ MByte of shared memory.\footnote{For example, the ``deep-buffered'' switches 
%like Cisco Catalyst 4500, 4700, 4900 and 4948 series have $16$ Mbyte shared memory, while the shallow-buffered 
%switches like Broadcom Triumph and Scorpion have $4$ MByte shared memory.} Suppose we imposes a 
%hard upper limit of $10\%$ shared memory for storing flow state, i.e. $409$ KByte out of $4$ MByte, we 
%will still be able to support $\sim$$21$,$000$ flows. 
%Indeed, in our simulation using the trace 
%from~\cite{imcDCmeasure}, the max memory consumption was merely $9.3$ KByte.}
%=======
%Today's switches are typically equipped with $4 - 16$ MByte of shared 
%high-speed memory (e.g., TCAM and SRAM).
%\footnote{For example, the ``deep-buffered'' switches
%like Cisco Catalyst 4500, 4700, 4900 and 4948 series have $16$ Mbyte shared memory,  while the shallow-buffered
%switches like Broadcom Triumph and Scorpion have $4$ MByte shared memory.}
%>>>>>>> .r8702

% matt: looks great. I just made a few minor tweaks
\cycomment{
Greenberg et al.~\cite{vl2} demonstrated that, under a production data center of a large scale cloud service, the 
number of concurrent flows going in and out of a machine is almost never more than $100$. 
Under a pessimistic scenario where \emph{every} server concurrently sends or receives $100$ flows, we have an average of $12$,$000$ active flows at each switch in a VL2 network (assuming flow-level equal-cost multi-path forwarding and $24$ $10$-Gbps Ethernet ports for each switch, the same as done in~\cite{vl2}).
Today's switches are typically equipped 
with $4 - 16$ MByte of shared high-speed memory\footnote{For example, the ``deep-buffered'' switches 
like Cisco Catalyst 4500, 4700, 4900 and 4948 series have $16$ MByte shared memory, while shallow-buffered 
switches like Broadcom Triumph and Scorpion have $4$ MByte shared memory~\cite{dctcp}.}, while storing all these flows requires $0.23$ MByte, only $5.72\%$ of a $4$ MByte shared memory.
Indeed, in our simulation using the trace 
from~\cite{imcDCmeasure}, the maximum memory consumption was merely $9.3$ KByte.}

Still, suppose our memory imposes a hard upper limit $M$ on the number of flows the switch can store. \sys, as described so far, will cause under-utilization when $\kappa > M$ and there are paused flows wanting to send.
In this underutilized case, we run an RCP~\cite{rcp} rate controller---which does not require per-flow state---alongside \sys. We inform RCP that its maximum link capacity is the amount of capacity not used by \sys, and 
we use RCP only for the less critical flows (outside the $M$ most critical) that are not paused by any other switches. RCP will let \emph{all} these flows run simultaneously using the leftover bandwidth.  Thus, even in this case of large $\kappa$ (which we expect to be rare), the result is simply a partial shift away from optimizing completion time and towards traditional fair sharing.

\subsubsection{The Flow Controller}
\label{sec:algorithm:switch:fc}
The flow controller performs Algorithm~\ref{alg:data} and~\ref{alg:ack} whenever it 
receives a data packet and an ACK packet, respectively.
The flow controller's objective is to \emph{accept} or \emph{pause} the flow.
A flow is accepted if \emph{all} switches along the path accept it. 
However, a flow is paused if \emph{any} switch pauses it.
This difference leads to the need for different actions:

\paragraphb{Pausing:} If a switch decides to pause a flow, it 
simply updates the ``pauseby'' field in the header (\pausebyh) to its ID. 
This is used to inform other switches and the sender that the flow should be paused.
%Therefore, a switch simply bypasses a packet paused by another switch.
Whenever a switch notices that a flow is paused by another switch, it removes the flow information from its state.
This can help the switch to decide whether it wants to accept other flows.

\paragraphb{Acceptance:}~To reach consensus across switches, flow acceptance takes two phases: (i) in 
the forward path (from source to destination), the switch computes the available bandwidth based on 
flow criticality (Algorithm~\ref{alg:delta}) and updates the rate and pauseby fields in the scheduling header;
%if a switch decides to accept a flow, it updates the rate field (\rateh) 
% currently paused by itself, it sets the pauseby field in the header to empty (\o); 
%similarly, if a switch decides to accept a flow which is currently sending data, it simply bypasses the packet. 
(ii) in the reverse path, if a switch sees 
\matt{an empty pauseby field in}
 %pauseby field with empty (\o) in 
the header, it updates 
the \emph{global} decision of acceptance to its state (\pausebyi and \ratei).

%Now we propose several \matt{extensions} to optimize our design:
\matt{We now} propose several optimizations to refine our design:

\paragraphb{Early Start:} Given a set of flows that are not paused by other switches, the switch accepts flows according 
to their criticality until the link bandwidth is fully utilized and the remaining flows are paused.
Although this ensures that the more critical flows can preempt other flows to fully utilize the link bandwidth, this can lead 
to \emph{low link utilization when switching between flows}.
To understand why, consider two flows, A and B, competing for a link's bandwidth. 
Assume that flow A is more critical than flow B. Therefore, flow A is accepted to occupy the entire link's bandwidth, while 
flow B is paused and sends only probe packets, e.g., one per its RTT.
By the time flow A sends its last packet (TERM), the sender of flow B does not know it should 
start sending data because of the feedback loop delay.
In fact, it could take one to two RTTs before flow B can start sending data. 
Although the RTT in data center networks is typically very small (e.g., $\sim$$150~\mu$s), the high-bandwidth 
short-flow nature makes this problem non-negligible. 
In the worst case where all the flows are short control messages ($<$$10$ KByte) that could finish in just one RTT, 
links could be idle more than half the time.

To solve this, we propose a simple concept, called \emph{Early Start}, to provide seamless flow switching. 
The idea is to start the next set of flows slightly before the current sending flows finish. 
Given a set of flows that are not paused by other switches, 
a \sys switch classifies a currently sending flow as nearly completed if the flow will finish 
sending in $K$ RTTs (i.e., \retimei $<$ $K$ $\times$ \rtti), for some small constant $K$.
We let the switch additionally accept as many nearly-completed flows as possible according to their criticality and subject to 
the resource constraint: aggregated flow transmission time (in terms of its estimated RTT) of the 
accepted nearly-completed flows ($\sum_i$\retimei{}$/$$RTT$$_i$) is no larger than $K$.
%We let the switch additionally accept $J$ flows such that (a) each of the flows will finish sending in $K$ RTTs and its index in the list is smaller than 
%, (2) it is more critical than any
%which will finish their sending in $K$ RTTs, such that $J$ 
%is maximized while the total expected transmission time (in terms of per-flow estimated RTT) of these $J$ flows is 
%smaller than a constant $K$.
%, i.e., $\sum_{i=1}^{J-1}\textrm{\retimei}/\textrm{RTT}_i \le K$.
The threshold $K$ determines how early and how many flows will be considered as nearly-completed. 
Setting $K$ to $0$ will prevent concurrent flows completely, resulting in low link utilization. 
Setting $K$ to a large number will result in congested links, increased queue sizes, and increased completion times of the most critical flows.
Any value of $K$ between $1$ and $2$ is reasonable, as the control loop delay is one RTT and the inter probing time is another RTT.
%it usually takes one to two RTTs until the feedback returns to the next sender.
In our current implementation we set $K=2$ to maximize the link utilization, and we use the rate controller to drain the queue. 
Algorithm~\ref{alg:delta} describes this in pseudocode, and we will show that Early Start provides seamless flow switching (\S\ref{sec:evaluation}).

\paragraphb{Dampening:}~When a more critical flow arrives at a switch, \sys will pause the current flow and switch to the new flow.
However, bursts of flows that arrive concurrently are common in data center networks, \matt{and} can potentially cause frequent flow switching, 
resulting in temporary instability in the switch state.\footnote{We later show that \sys can quickly 
converge to the equilibrium state when the workload is stable (\S\ref{sec:properties}).} To suppress this, we 
use dampening: after a switch has accepted a flow, it can only accept other paused 
flows after a given small period of time, as shown in Algorithm~\ref{alg:data}.

\paragraphb{Suppressed Probing:} One could let a paused sender send one probe per RTT. However, this can 
introduce significant bandwidth overhead because of the \cycomment{small RTTs} in data center networks. 
For example, assume a $1$-Gbps network where flows have an RTT of $150~\mu$s. A paused flow that 
sends a $40$-byte probe packet per RTT 
consumes $\frac{40~\textrm{Byte}}{150~\mu\textrm{s}}/1~\textrm{Gbps} \approx 2.13\%$ of the total bandwidth.
The problem becomes more severe with larger numbers of concurrent flows.
%a large number of concurrent flows in data center networks.

To address this, we propose a simple concept, called \emph{Suppressed Probing} to reduce the probing frequency. 
We make an observation that only the paused flows that are about to start have a need to send probes frequently.
Therefore, it is desirable to control probing frequency based on the 
flow criticality and the network load. 
To control probing frequency, one would need to estimate {\em flow waiting time} 
(i.e., how long does it take until the flow can start sending).
Although it is considered hard to predict future traffic workloads in data centers, 
switches can easily estimate a lower bound of the flow waiting time by checking their flow list.
Assuming each flow requires at least $X$ RTTs to finish, a \sys switch estimates that a flow's waiting time is 
at least $X \times \max_{\forall\ell}\{\textrm{Index}(\ell)\}$ RTTs, where 
Index$(\ell)$ is the flow index in the list on link $\ell$.
The switch sets the inter-probing time field (\pfh) to $\max\{\textrm{\pfh}, X \times\textrm{Index}(\ell) \}$ in the scheduling 
header to control the sender probing rate (\pfs), as shown in Algorithm~\ref{alg:ack}.
The expected per-RTT probing overhead is significantly reduced from $O(n)$ ($n$ flows, 
each of which sends one probe per RTT) to $\frac{1}{X}\sum_{k=1}^{n}{1/k} = O(\log n)$.
In our current implementation, we conservatively set $X$ to $0.2$ RTT$_{avg}$. 

\begin{algorithm}[!h]
%\algfontsize
\footnotesize
    \If{\pausebyh = other switch}
    {
        Remove the flow from the list if it is in the list; \Return; 
    }
    \If{the flow is not in the list}
    {
        \If{the list is not full or the flow criticality is higher than the least critical flow in the list}
        {
           Add the flow into the list with rate \ratei = 0. Remove the least critical flow from the list whenever the list has more than $\kappa$ flows.
        }
        \Else
        {
           %\pausebyh = myID; \codecomment{// Pause the flow}\\
           Set \rateh to RCP fair share rate;\\
           \lIf{\rateh = $0$}{\pausebyh = myID;}\\
           \Return;\\
        }
    }
    Let $i$ be the flow index in the list; Update the flow information: <\deadlinei{},\retimei{},$RTT_i$> = <\deadlineh{},\retimeh{},$RTT_H$>;\\
    
    %{
    %    \codecomment{// Other router paused it, bypass it}\\
    %    \pausebyi = \pausebyh; \Return;     
    %}
    %\If{$i \le Delta()$} %binary version
    \If{$W$$=$$\min(Availbw(i),$\rateh{}$)$$>$$0$} 
    {
        \If{the flow is not sending (\pausebyi $\ne$ \o), and the switch just accepted another non-sending flow}
        {
	    \pausebyh = myID; \pausebyi = myID; \codecomment{// Pause it}
        }
        \lElse
        {
            \pausebyh = \o; \rateh{}$=$$W$; \codecomment{// Accept it}
            
        }
    }
    \lElse
    {
        \pausebyh = myID; \pausebyi = myID; \codecomment{// Pause it}
    }
\caption{\sys Receiving Data Packet}
\label{alg:data}
\end{algorithm}
\algtotext

\begin{algorithm}[!h]
%\algfontsize
\footnotesize
   $X$=0; $A$=0\;
   \For{$(i=0;~i<j;~i=i+1)$}
   {
      %\If{\pausebyi is not other router}
      %{
         \If{\retimei{}/$RTT_i$ < $K$~and~$X$ < $K$}
         {
             $X=X$ + \retimei/$RTT_i$\;
         }
         \Else
         {
	     $A=A$ + \ratei\;
         }
      %}
      \lIf{$A \ge C$}{\Return $0$\;}
   }
   \Return $C - A$\;
\caption{Availbw($j$)}
\label{alg:delta}
\end{algorithm}
\algtotext

%Binary Version
%\begin{algorithm}
%   $X$=0;\\
%   \For{$(i=0;~i<ListSize;~i=i+1)$}
%   {
%      \If{\pausebyi is not other router}
%      {
%         $X = X + $\retimei;
%      }
%      \lIf{$X > K$}{break;}
%   }
%   \Return $i$;
%   \caption{Delta()}
%\label{alg:delta}
%\end{algorithm}

\begin{algorithm}[!h]
%\algfontsize
\footnotesize
    \If{\pausebyh = other switch}
    {
        Remove the flow from the list if it is in the list;
    }
    \If{\pausebyh $\neq$ \o} 
    {
        \rateh = $0$; \codecomment{// Flow is paused}
    } 
    \If{the flow is in the list with index $i$}{\pausebyi = \pausebyh; \pfh = $\max$$\{$\pfh, $X \times i$$\}$; \ratei = \rateh;}
\caption{\sys Receiving ACK}
\label{alg:ack}
\end{algorithm}
\vspace{12pt}
%\algtotext

\subsubsection{The Rate Controller}
\label{sec:algorithm:switch:rc}
The rate controller's objective is to control the aggregated flow sending rate of the flows accepted by the flow controller 
based on the queue size and the measured aggregate traffic.
The rate adjustment serves the following purposes. 
First, whenever the queue builds up due to the use of Early Start, it helps drain the queue right after flow switching.
Second, it helps tolerate the congestion caused by transient inconsistency. 
For example, if a packet carrying the pausing information gets lost, the 
corresponding sender that is supposed to stop will still be sending, and the rate 
controller can reduce the flow sending rate to react to the congestion.
Finally, this allows \sys to be friendly to other transport protocols in a multi-protocol network. 

The rate controller maintains a single variable $C$ to control the aggregated flow sending rate. 
This variable will be used to compute the sending rate field (\rateh) in the scheduling header, as shown in Algorithm~\ref{alg:delta}.

The rate controller updates $C$ every $2$ RTTs because of the feedback-loop delay: we need about one RTT latency 
for the adjusted rate to take effect, and one additional RTT is used to measure the link congestion with that newly adjusted sending rate.

The rate controller updates $C$ to $\max\{0,r_{\textrm{\sys}} - q / (2\times \textrm{RTT})\}$, where $q$ is the 
instantaneous queue size and $r_{\textrm{\sys}}$ is the per-link aggregated rate for \sys flows.
If all traffic is transported using \sys, one can configure the $r_{\textrm{\sys}}$ to be equal to the link rate. 
This allows \sys flows to send with its highest possible rate.
Otherwise, the network administrator can decide their priority by setting $r_{\textrm{\sys}}$ accordingly. 
For example, one could give preference to other protocols by periodically updating $r_{\textrm{\sys}}$ to the difference 
between the link rate and the measured aggregated traffic of the other protocols.
Alternatively, one could set it based on the per-protocol traffic amount to achieve fairness across protocols.

%% file: properties.tex
\section{Formal Properties}
\label{sec:properties}

In this section, we present two formal properties of \sys{} --- deadlock-freedom and finite convergence time. 

\paragraphb{Assumptions:} Without loss of generality, we assume there is no 
packet loss. Similarly, we assume flows will not be paused due to the use of flow dampening. 
Because \sys flows periodically send probes, the properties we discuss in this section will hold with additional latency when 
the above assumptions are violated. For simplicity, we also assume the link rate $C$ is equal to 
the maximal sending rate \maxrates (i.e., \schrates $=$ $0$ or $C$). Thus, each link accepts only one flow at a time.

\paragraphb{Definitions:} 
%Recall that \sys switches share a common flow comparator algorithm to determine whether a flow is more critical than another (\S\ref{sec:algorithm:switch}).
We say a flow is \emph{competing} with another flow if and only if they share at least one common link. 
Moreover, we say a flow $F_1$ is a \emph{precedential} flow of flow $F_2$ if and only if they are 
competing with each other and flow $F_1$ is more critical than flow $F_2$. 
We say a flow $F$ is a \emph{driver} if and only if (i) flow $F$ is more critical than any other competing flow, or (ii) all the 
competing flows of flow $F$ that are more critical than flow $F$ are non-drivers.

%\paragraphb{Results (proof is in~\cite{ourtr}):} \cycomment{We verify that \sys has no \emph{deadlock}, which is a situation where two or 
%more competing flows are paused and are each waiting for the other to finish (and therefore neither ever does).
%We further prove that \sys will converge to the \emph{equilibrium} in $P_{\max}+1$ RTTs for stable 
%workloads, where $P_{\max}$ is the maximal number of precedential flows of any flow.
%Given a collection of active flows, the equilibrium is defined as a state where 
%the drivers are accepted while the remaining flows are paused.
%}

\paragraphb{Results:} \cycomment{In Appendix A, we verify that \sys has no \emph{deadlock}, which is a situation where two or 
more competing flows are paused and are each waiting for the other to finish (and therefore neither ever does).
In Appendix B, we further prove that \sys will converge to the \emph{equilibrium} in $P_{\max}+1$ RTTs for stable 
workloads, where $P_{\max}$ is the maximal number of precedential flows of any flow.
Given a collection of active flows, the equilibrium is defined as a state where 
the drivers are accepted while the remaining flows are paused.
}

\newtheorem{theorem}{Property}
\newtheorem{lemma}{Lemma}

\cut{

\begin{theorem}
\label{theorem:deadlock}
\vspace{-5pt}\textnormal{
\sys has no deadlock.
}
\end{theorem}
\begin{proof}
\vspace{-5pt}See Appendix A.
\cut{Deadlock is a situation where two or more competing flows are paused and are each waiting for the other to 
finish, and therefore neither ever does. We verify \sys has no deadlock by showing that a necessary condition of 
deadlock, Hold and Wait, is false. 
To reach this, we need to show that there is no \emph{partial acceptance}: a flow 
is accepted by some intermediate switches, while paused by other switches along the routing path.
In \sys, a flow is accepted \emph{only} after \emph{every} switch along the path accepts the flow. 
Moreover, if a \sys flow is paused, the switch who pauses this
flow will update the pauseby field in the scheduling header (\pausebyh) to its ID.
This information will eventually reach all the other switches along the path, as 
even the paused flows would send probes periodically.
%Even if packet loss happens, the same switch will update the blockby field on the next packet 
%sent by the sender of the paused flow.
Whenever a switch notices that a flow is paused by another switch, it will not consider accepting this flow.
Thus, a paused flow will not be accepted by \emph{any} switch along the path.}
\end{proof}
}

%\mattc{I think this isn't sufficient to prove there is no deadlock. You demonstrate
%that the flows themselves can't deadlock, but I think you also need to show
%that the control messages don't deadlock either. In more detail: I think
%there are two parts to the flow setup. There's a "reservation" step where
%a flow contacts each switch to see if it can be admitted. Switches do some
%local state updates to remember if they accepted the flow. Then the actual
%flow gets sent. How do you know the reservation step can't deadlock? }
%\cy{Every packet serves as a request packet, so there's no reservation step in our current algorithm. Does that answer the question?}
%\mattc{I think you are right. Just to double check if I understand this, is it true that your algorithm can temporarily oversubscribe a
%link? (eg imagine two flows starting up at nearly the same time -- your scheme may temporarily admit both but then quickly soon after tell one to pause) If so, then I understand things, and I think you are correct here.
%}\cy{Yes, our scheme may temporarily admit both and then quickly pause one in your example.}
% \mattc{Ok, then I think you are right here. I'm removing my comment}

\cut{
\begin{theorem}
\label{theorem:accept}
\label{theorem:pause}
\textnormal{
If all the precedential flows of a flow $F$ are paused (or it has none), flow $F$ will be accepted in at most one RTT.
If any precedential flow of a flow $F$ is accepted, flow $F$ will be paused in at most one RTT.
%If all the precedential flows of a flow $F$ are sending at their scheduling rate 
%(or it has none), then flow $F$ will send at the scheduling rate in at most two RTTs.
}
\end{theorem}
%\begin{proof}
%In this case, flow $F$ is more critical than any non-paused competing flow. 
%Although \sys usually takes only one RTT for switches to accept flow $F$, in the worst case it could take up to two RTTs:
%In the first RTT, flow $F$ might be temporarily blocked by an intermediate switch due to flow dampening. 
%This can happen if another flow is accepted by an intermediate switch right before the first packet of flow $F$ arrives at the switch.
%However, after the first RTT, all the intermediate switches will store flow $F$'s information.
%Therefore, flow F will be accepted in the second RTT. Flow dampening cannot happen in
%the second RTT because switches will not accept other flows after it has stored flow $F$'s information.
%\end{proof}
}

\cut{
\begin{proof}
\prooffontsize
This property follows directly from the \sys flow controller algorithm.
\end{proof}
}

\cut{
\begin{theorem}
\label{theorem:converge}\vspace{-7pt}
\textnormal{
\sys will converge to the equilibrium state in $P_{\max}+1$ RTTs for stable 
workloads, where $P_{\max}$ is the maximal number of precedential flows of any flow.
Given a collection of active flows, the equilibrium is defined as a state where 
the drivers are accepted while the remaining flows are paused.
}
\end{theorem}
\begin{proof}\vspace{-5pt}
See Appendix B.
\end{proof}
}

\cut{
\begin{proof}
\prooffontsize
We show that when the workload is stable (no new flows arrive and no sending flow finishes), a flow will 
be accepted if it is a driver and will otherwise be paused in $P + 1$ RTTs, where $P$ is the number of its precedential flows.
We prove this will hold for any flow $F$ that is the $m$-th critical flow in the network by induction on $m$.
%, the number of its competing flows that are more critical than flow $F$.
When $m=1$, flow $F$ is a driver according to Statement A. %and is the most critical flow compared with its competing flows. 
Thus, it will be accepted in one RTT according to Property~\ref{theorem:accept}. 
When $m=n+1$, there exist $n$ flows $F_1 \cdots F_n$ that are more critical than flow $F$. 
Without loss of generality, out of these $n$ flows, we assume there are $n' \le n$ precedential flows (as they are competing with flow $F$).
Suppose that flow $F$ is a driver. Then, according to Statement B, all these $n'$ flows are non-drivers. 
By the induction hypothesis, these $n'$ competing flows will all be paused in $P'+1$ RTTs, where $P'$ is the maximal possible number of precedential flows
of these $n'$ flows. As each of these $n'$ flows will have at most $n-1$ precedential flows, we have $P' \le n-1$.
After these flows are paused, it takes at most an RTT for switches to accept flow $F$ according to 
Property~\ref{theorem:accept}, and therefore the flow $F$ will be 
accepted in $(P'+1)+1 \le n+1$ RTTs. 
Suppose now that flow $F$ is not a driver. According to Statement B, among those $n'$ competing flows, there exists $>$$0$ drivers. 
Similarly, among these drivers, the maximal number of precedential flows of each driver is at most $n-1$. 
By the induction hypothesis, these drivers will be accepted in at most $n$ RTTs, and 
after this, flow $F$ will be paused in one RTT according to Property~\ref{theorem:pause}. 
%When a driver arrives, we usually need one RTT latency before the sender can start sending. 
%After switches accept the driver, we need at most one RTT for the other non-driver flows 
%to stop sending after receiving the updated feedback.
%Thus, \sys will converge in at most $3$ RTTs when there is no packet loss. 
\end{proof}
}

%% file: evaluation.tex
\section{\sys Performance}
\label{sec:evaluation}

%\matt{
%In this section, we evaluate \sys's performance through
%comprehensive simulations.
%We first describe our evaluation settings~(\S\ref{sec:perf:setting}).
%We then evaluate \sys under a ``query aggregation'' scenario~(\S\ref{sec:perf:       aggregation}).
%Then, we evaluate \sys under a wider range of workloads~(\S\ref{sec:perf:workload},
%\ref{sec:perf:dynamic}, and \ref{sec:perf:scale}).
%Finally, we show that \sys is resilient to inaccurate flow information and
%packet loss (\S\ref{sec:perf:badcases}).
%}

%cy: Let's highlight our results. This evaluation section is pertty huge (5 pages) and we need a good roadmap here to help reader understand 
%what did we do...
% matt: yeah, good point, maybe I was a bit too aggressive at cutting to save space. I like your version better, let's keep that.

\cycomment{In this section, we evaluate \sys's performance through comprehensive simulations. 
We first describe our evaluation setting~(\S\ref{sec:perf:setting}).
Under a ``query aggregation'' scenario, \sys achieves near-optimal 
performance and greatly outperforms \bntl, RCP and TCP~(\S\ref{sec:perf:aggregation}).
We then demonstrate that \sys retains its performance gains under 
different workloads, including two realistic data center workloads from measurement studies (\S\ref{sec:perf:workload}), followed 
by two scenarios to demonstrate that \sys does not compromise on traditional congestion control performance metrics~(\S\ref{sec:perf:dynamic}).
Moreover, \sys retains its performance benefits on a variety of data center topologies (Fat-Tree, BCube and Jellyfish) and provides 
clear performance benefits at all scales that we evaluated (\S\ref{sec:perf:scale}). 
Further, we show that \sys is highly resilient to inaccurate flow information and packet loss (\S\ref{sec:perf:badcases}).}

%that \sys outperforms state-of-the-art transport protocols.

\subsection{Evaluation Setting}
\label{sec:perf:setting}

Our evaluation considers two classes of flows:

\paragraphb{Deadline-constrained Flows} are time sensitive flows that have specific deadline \matt{requirements} to meet. 
The flow size is drawn from the interval [$2$ KByte, $198$ KByte] using a uniform distribution, as done in a prior study~\cite{bntl}.
This represents query traffic ($2$ to $20$ KByte in size) and delay sensitive short messages ($>$$100$ KByte) in data center networks~\cite{dctcp}.
The flow deadline is drawn from an exponential distribution with mean $20$ ms, as suggested by~\cite{bntl}. 
\cycomment{However, some flows could have tiny deadlines that are unrealistic in real network applications.
To address this, we impose a lower bound on deadlines, and we set it to $3$ ms in our experiments.}
We use \emph{Application Throughput}, the percentage of flows that meet their deadlines, as the performance metric of deadline-constrained flows.

\paragraphb{Deadline-unconstrained Flows} are flows that have no specific deadlines, but it is desirable that they finish early.
For example, Dryad jobs that move file partitions across machines. Similarly, we assume the flow size 
is drawn uniformly from an interval with a mean of $100$/$1000$ KByte.
We use the average flow completion time as the performance metric.

We have developed our own event-driven packet-level simulator written in C++. The simulator models the following schemes:

\paragraphb{\sys:} We consider different variants of \sys. 
We use {\sys{}(Full)} to refer to the complete version of \sys, including Early Start (ES), Early Termination (ET) and Suppressed Probing (SP).
Likewise, we refer to the partial version of \sys which excludes the above three algorithms as {\sys{}(Basic)}.
To better understand the performance contribution of each algorithm, we further extend {\sys{}(Basic)} to {\sys{}(ES)} and {\sys{}(ES+ET)}.
%However, Early Termination (ET) is not applicable to deadline-unconstrained flows.

\paragraphb{\bntl:} We implemented a complete version of \bntl~\cite{bntl}, including the rate request processing 
procedure, the rate adaptation algorithm (with the suggested 
parameters $\alpha=0.1$ and $\beta=1$), and the quenching algorithm.
In the original algorithm when the total demand exceeds the switch capacity, the fair share rate becomes negative. 
\cycomment{We found this can cause a flow to return the allocated bandwidth it already reserved, resulting in unnecessarily missed deadlines. 
Therefore, we add a constraint to enforce the fair share bandwidth $fs$ to always be non-negative, which improves \matt{\bntl's performance}.
%the performance of \bntl.
}  %matt: nice!

\paragraphb{RCP:} We implement RCP~\cite{rcp} and optimize it by counting the exact number of flows at switches. \cycomment{We found this 
improves the performance by converging to the fair share rate more quickly, significantly reducing 
the number of packet drops when encountering a sudden large influx of new flows~\cite{rcpproblem}.} This 
is exactly equivalent to \bntl when flows have no deadlines.

\paragraphb{TCP:} We implement TCP Reno and optimize it by setting a small $RTO_{\min}$ to alleviate the TCP 
Incast problem, as suggested by previous studies~\cite{dctcp,incast09}.

Unless otherwise stated, we use single-rooted tree, a commonly used data center topology 
for evaluating transport protocols~\cite{dctcp,incast09,bntl, ictcp}.
In particular, our default topology is a two-level $12$-server single-rooted tree topology with $1$ Gbps link rate (Figure~\ref{fig:topology2}), 
the same as used in \bntl. 
\matt{We vary the traffic workload and topology in \S\ref{sec:perf:workload} and \S\ref{sec:perf:scale}.}

%We further extend the workload and topology to more complex cases in \S\ref{sec:perf:workload} and \S\ref{sec:perf:scale}.

%\mattc{I think you should mention the default sending pattern you use too here}   

\begin{figure}[t]
\centering
\hspace{5pt}\subfloat[]{\label{fig:topology2} \includegraphics[width=1.3in]{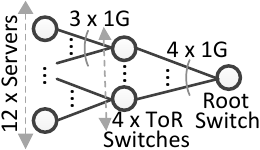} }\hspace{5pt}
\subfloat[]{\label{fig:topology1} \includegraphics[width=1.44in]{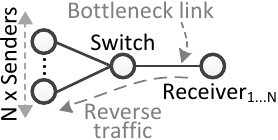} }\\\objtolbl
\caption{\figtitle{Example topologies: (a) a $17$-node single-rooted tree topology; (b) a single-bottleneck topology: sending servers associated with different
flows are connected via a single switch to the same receiving server. 
\cready{Both topologies use $1$ Gbps links, a switch buffer of $4$ MByte, and FIFO tail-drop queues. 
Per-hop transmission/propagation/processing delay is set to $11$/$0.1$/$25$ $\mu$s.}}}
\label{fig:topology}
\end{figure}

\subsection{Query Aggregation}
\label{sec:perf:aggregation}

In this section, we consider a scenario called \emph{query aggregation}: a number of senders 
initiate flows at the same time to the same receiver (the aggregator).
This is a very common application scenario in data center networks and has been adopted by a number of previous works~\cite{ictcp,bntl,dctcp}. 
We evaluate the protocols in both the deadline-constrained 
case~(\S\ref{sec:perf:aggregation:deadline}) and the deadline-unconstrained 
case~(\S\ref{sec:perf:aggregation:nodeadline}).

\subsubsection{Deadline-constrained Flows}
\label{sec:perf:aggregation:deadline}

\begin{figure}[!t]
\centering
\subfloat[]{ \label{fig:burst:nflow}   \includegraphics[width=3.2in]{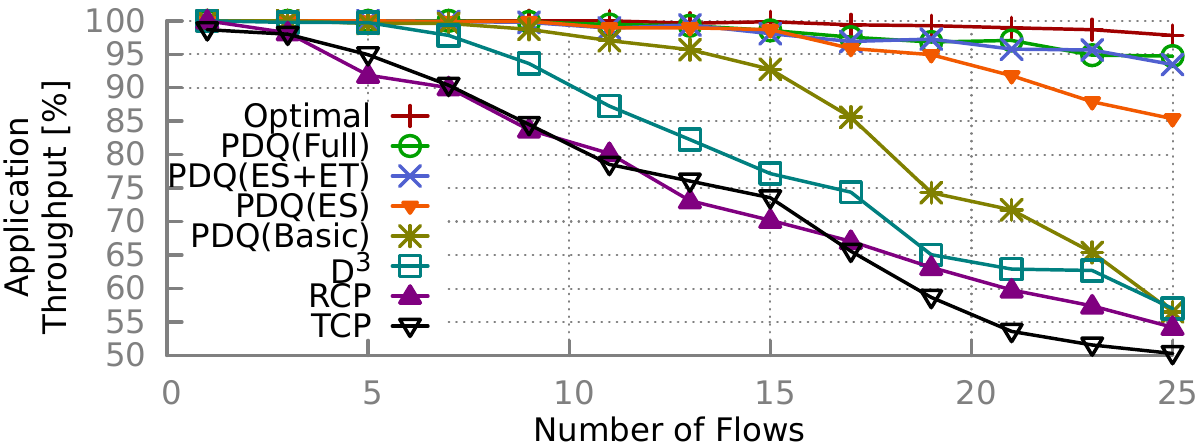}}\\\objtoobj
\subfloat[]{ \label{fig:burst:flowsize}\includegraphics[width=3.2in]{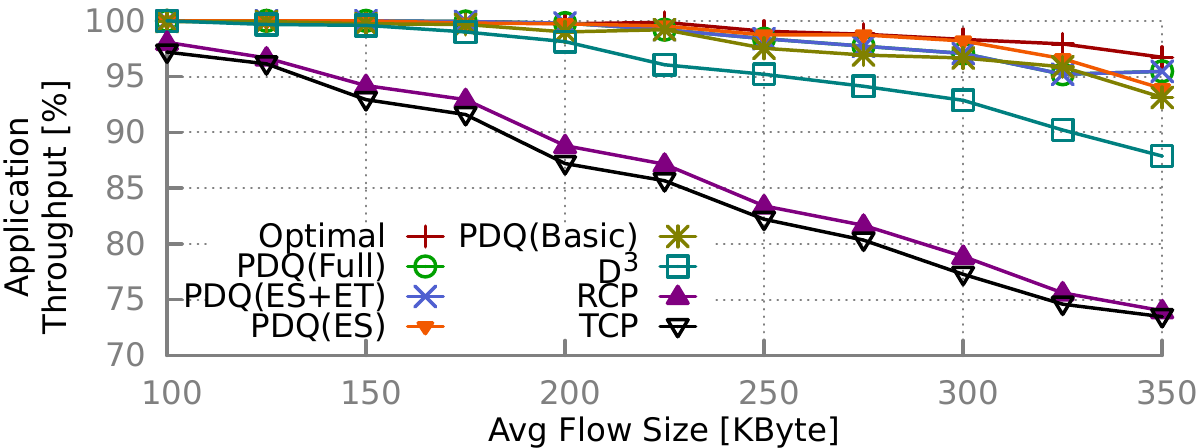}}\\\objtoobj
\subfloat[]{ \label{fig:burst:deadline}\includegraphics[width=3.2in]{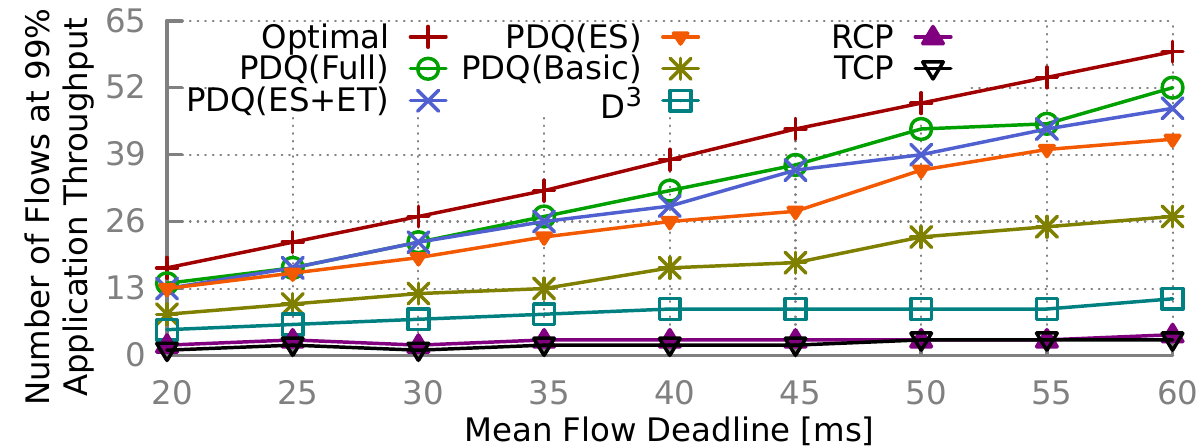}}\\\objtoobj
\subfloat[]{ \label{fig:burst:nflow:ND}   \includegraphics[width=3.2in]{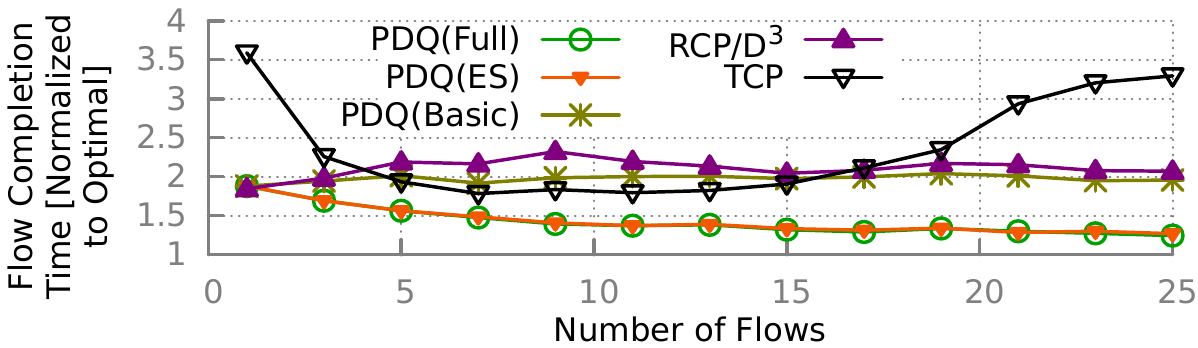}}\\\objtoobj
\subfloat[]{ \label{fig:burst:flowsize:ND}\includegraphics[width=3.2in]{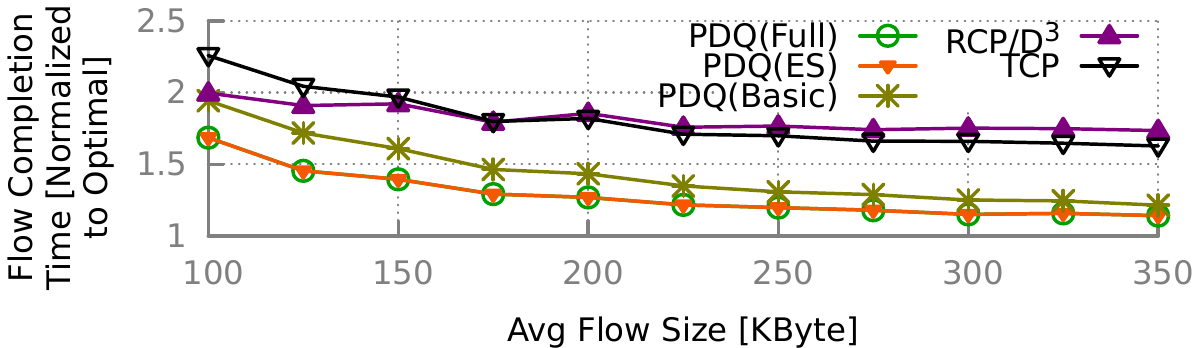}}\\\objtolbl
%\subfloat[]{ \label{fig:plall}\includegraphics[width=2.2in]{figures/plall}}
\caption{\figtitle{\sys outperforms \bntl, RCP and TCP and achieves near-optimal 
performance. Top three figures: deadline-constrained flows; bottom two figures: deadline-unconstrained flows.
\cut{Flow burst scenario2. Two-level single-rooted tree. Aggregation traffic. (a, b, c) Deadline-constrained flow. (d, e) Deadline-unconstrained flows, we plot the average flow completion time. (a, c) Varying the number of flows. (b, d) Varying the average flow size. 
}}}
\label{fig:burst}
\end{figure}

\paragraphb{Impact of Number of Flows:} 
We start by varying the number of flows.\footnote{We randomly assign $f$ flows to $n$ senders while 
ensuring each sender has either $\lfloor f/n \rfloor$ or $\lceil f/n \rceil$ flows.} To understand bounds on 
performance, we also simulate an {\em optimal} solution, where an omniscient scheduler can
control the transmission of any flow with no delay.
%To derive the optimal solution, we assume an omniscient scheduler that can control the transmission of any flow with no delay. 
It first sorts the flows by EDF, and then \matt{uses} a dynamic programming algorithm 
to discard the minimum number of flows that cannot meet their deadlines (Algorithm $3.3.1$ in \cite{pinedo}).
We observe that \sys has near-optimal application throughput across a wide range of loads (Figure~\ref{fig:burst:nflow}). 

Figure~\ref{fig:burst:nflow} demonstrates that Early Start is very effective for short flows. 
By contrast, \sys{}(Basic) has much lower application throughput, especially during heavy system 
load because of the long down time between flow switching.
Early Termination further improves performance by discarding flows that cannot meet their deadline. 
Moreover, Figure~\ref{fig:burst:nflow} demonstrates that, as the number of concurrent 
flows increases, the application throughput of \bntl, RCP and TCP decreases significantly.

\paragraphb{Impact of Flow Size:} 
We fix the number of concurrent flows at $3$ and study the impact of increased flow size on the application throughput.
Figure~\ref{fig:burst:flowsize} shows that as the flow size increases, the performance of deadline-agnostic schemes (TCP and RCP) degrades 
considerably, while \sys remains very close to optimal regardless of the flow size. 
However, Early Start and Early Termination provide fewer benefits in this scenario 
because of the small number of flows.

\paragraphb{Impact of Flow Deadline:} 
Data center operators are particularly interested in the \matt{operating regime} where 
the network can satisfy almost every flow deadline.
To this end, we attempt to find, using a binary search procedure, 
the maximal number of flows a protocol can support while ensuring $99\%$ application throughput.
We also vary the flow deadline, which is drawn from an exponential distribution, to observe the 
system performance with regard to different flow deadlines with mean between $20$ ms to $60$ ms.
Figure~\ref{fig:burst:deadline} demonstrates that, compared with \bntl, \sys can 
support $>$$3$ times more concurrent flows at $99\%$ application throughput, and this 
ratio becomes larger as the mean flow deadline increases.
Moreover, Figure~\ref{fig:burst:deadline} shows that Suppressed Probing becomes more useful as the number of concurrent flows increases. 

\subsubsection{Deadline-unconstrained Flows}
\label{sec:perf:aggregation:nodeadline}

\paragraphb{Impact of Flow Number:} 
For deadline-unconstrained case, we first measure the impact of the number of flows on the average flow completion time.
Overall, Figure~\ref{fig:burst:nflow:ND} demonstrates that \sys can effectively approximate the optimal flow completion time. 
The largest gap between \sys and optimal happens when there exists only one flow and is due to flow initialization latency. 
RCP has a similar performance for the single-flow case. However, its flow completion time 
becomes relatively large as the number of flows increases. 
TCP displays a large flow completion time when the number of flows is small due to the inefficiency of slow start.
When the number of concurrent flows is large, TCP also has an increased flow completion time due to the TCP incast problem~\cite{incast09}.
 
\paragraphb{Impact of Flow Size:} 
We fix the number of flows at $3$, and Figure~\ref{fig:burst:flowsize:ND} shows the flow 
completion time as the flow size increases. We demonstrate that \sys can better approximate optimal flow completion time as flow size increases.
The reason is intuitive: the adverse impact of \sys inefficiency (e.g., flow initialization latency) on flow completion time becomes 
relatively insignificant as flow size increases.

\begin{figure*}[t]
\centering
\subfloat[]{ \label{fig:trafficp}   \includegraphics[width=6.8in]{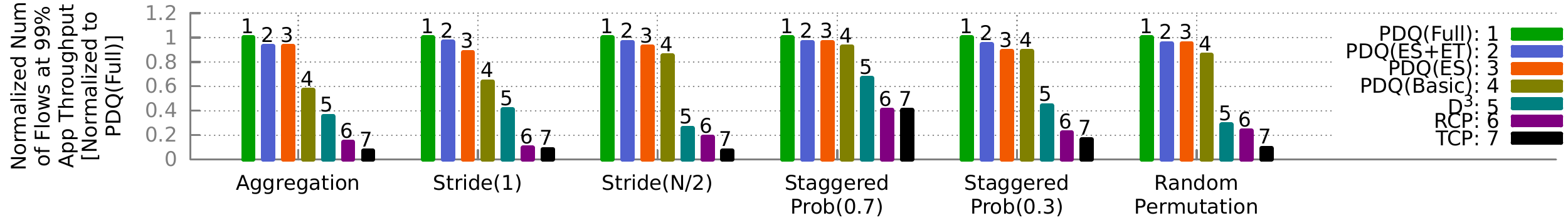}}\\\objtoobj
\subfloat[]{ \label{fig:trafficp:ND}   \includegraphics[width=6.8in]{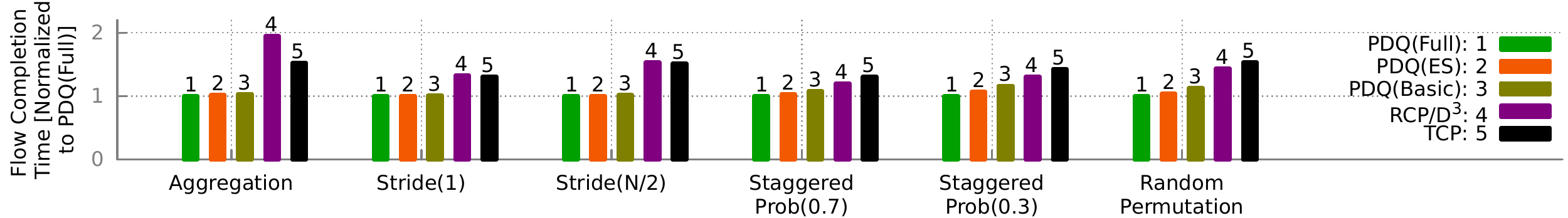}}\\\objtolbl
\caption{\figtitle{\sys outperforms \bntl, RCP and TCP across traffic patterns. (a) Deadline-constrained flows; (b) Deadline-unconstrained flows.}}
\label{fig:trafficpattern}
\end{figure*}

\subsection{Impact of \matt{Traffic} Workload}
\label{sec:perf:workload}

\paragraphb{Impact of Sending Pattern:}
We study the impact of the following sending patterns:  
%\mattc{
%This is the first time sending pattern is mentioned. However, you must
%have used some sending pattern in previous subsections. Which
%sending pattern did you use in previous subsections? It makes it hard to completely understand your other results in previous sections if I don't know the sending pattern. I think you should mention that when you talk about methodology at the start of this evaluation section} 
%\cycomment{[cy: Reader should be able to understand the sending pattern in Section 5.2. Check again the description in Sec 5.2. Let me know if you still think it's unclear.]}
%<<<<<<< .mine
\textbf{(i)} \emph{Aggregation}: multiple servers send to the same aggregator,
\matt{as done in the prior experiment}. 
%This is the one used in the prior experiment; 
%=======
%\textbf{(i)} \emph{Aggregation}: multiple servers send to the same aggregator. 
%This is the one used in the prior experiment; 
%>>>>>>> .r8705
\textbf{(ii)} \emph{Stride}($i$): a server with index $x$ sends to the host with 
index $(x+i)$ $mod$ $N$, where $N$ is the total number of servers; 
\textbf{(iii)} \emph{Staggered Prob}($p$): a server sends to another server under 
the same top-of-rack switch with probability $p$, and to any other server with probability $1-p$; 
\textbf{(iv)} \emph{Random Permutation}: a server sends to another randomly-selected server, with a constraint 
that each server receives traffic from exactly one server (i.e., $1$-to-$1$ mapping).

Figure~\ref{fig:trafficpattern} shows that \sys reaps its benefits across 
all the sending patterns under \matt{consideration}. 
The worst \matt{pattern} for \sys is the Staggered Prob($0.7$) due to the fact that 
the variance of the flow RTTs is considerably larger. In this \matt{sending} pattern, the non-local flows that 
pass through the core network could have RTTs $3-5$ times larger than the local flow RTTs. 
Thus, the \sys rate controller, whose update frequency is based on a measurement of \emph{average} flow RTTs, could 
slightly overreact (or underreact) to flows with relatively large (or small) RTTs. 
However, even in such a case, \sys still outperforms the other schemes considerably.

\paragraphb{Impact of \matt{Traffic Type}:}
We consider two workloads collected from real data centers.
First, we use a workload with flow sizes following the distribution from a 
large-scale commercial data center measured by Greenberg et al.~\cite{vl2}. 
It represents a mixture of long and short flows: 
\matt{Most flows are small, and most} of the delivered bits are contributed by long flows.
In the experiment, we assume that the short flows (with a size of $<$$40$ KByte) are deadline-constrained.
We conduct these experiments with random permutation traffic.

Figure~\ref{fig:vl2:mice} demonstrates that, under this particular workload, \sys outperforms the other protocols by supporting 
a significantly higher flow arrival rate. 
We observed that, in this scenario, \sys{}(Full) considerably outperforms \sys{}(ES+ET). 
This suggests that Suppressed Probing plays 
an important role in minimizing the probing overhead especially when there exists 
a large collection of paused flows. 
Figure~\ref{fig:vl2:elephant} shows that \sys has lower flow completion time for long flows: 
a $26\%$ reduction compared with RCP, and a $39\%$ reduction compared with TCP.

We also evaluate performance using a workload collected from a university data center with $500$ servers~\cite{imcDCmeasure}.
In particular, we first convert the packet trace, which lasts $10$ minutes, to flow-level summaries using Bro~\cite{bro}, then we fed it to the simulator.
Likewise, \sys outperforms other schemes in this regime (Figure~\ref{fig:imc:mixed}).

\begin{figure}[t]
\centering
\subfloat[]{ \label{fig:vl2:mice}   \includegraphics[width=3.2in]{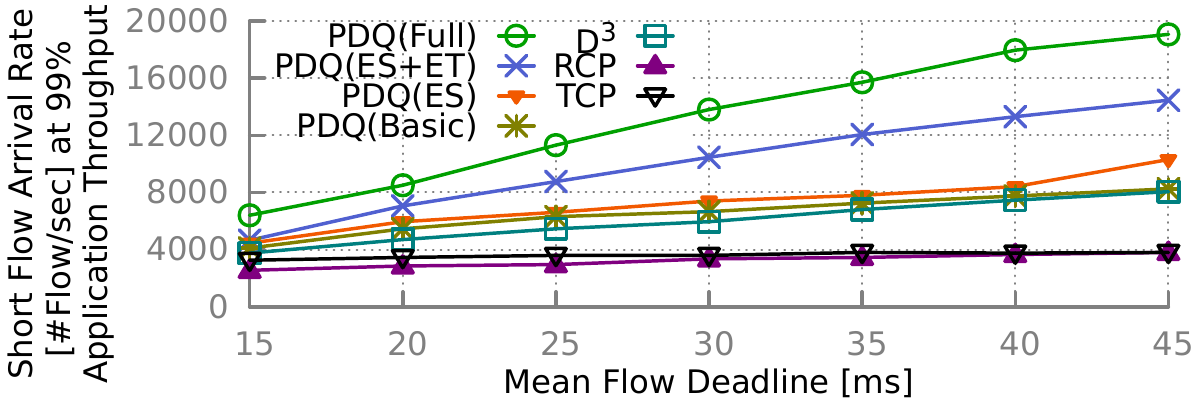}}\\\objtoobj
\hspace{-3pt}\subfloat[]{ \label{fig:vl2:elephant}   \includegraphics[width=1.4in]{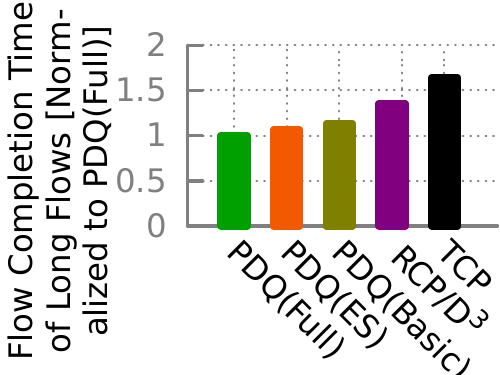}}\hspace{12pt}
\subfloat[]{ \label{fig:imc:mixed}   \includegraphics[width=1.4in]{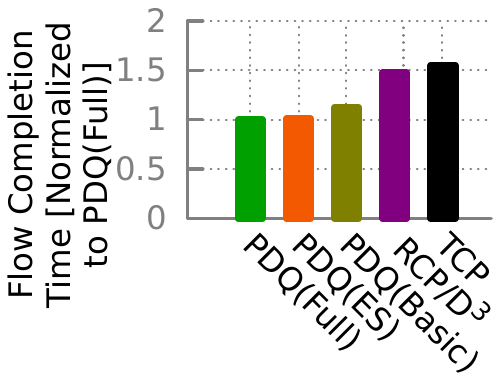}}\objtolbl
\caption{\figtitle{Performance evaluation under realistic data center workloads, collected from \emph{(a, b)} a 
production data center of a large commercial cloud service\cready{~\cite{vl2}} and \emph{(c)} a university data center located in Midwestern United States \cready{(EDU1 in \cite{imcDCmeasure})}.
}}
\label{fig:vl2}
\end{figure}

\subsection{\matt{Dynamics} of \sys}
\label{sec:perf:dynamic}

\matt{
%Next, we show \sys's performance under two kinds of traffic dynamics:
Next, we show \sys's performance over time through two scenarios, each
with varying traffic dynamics:
}

\paragraphb{Scenario \#1 (Convergence Dynamics):} Figure~\ref{fig:case1} shows that \sys provides seamless flow switching.
We assume five flows that start at time $0$. The flows have no deadlines, and each flow has a size of $\sim$$1$ MByte.
The flow size is perturbed slightly such that a flow with smaller index is more critical.
%\mattc{I can't parse this sentence -- how does changing the flow size
%make a flow more critical? I think there's a missing phrase here.}
%CY: See Section 3.4
Ideally, the five flows together take $40$ ms to finish because each flow requires 
a raw processing time of $\frac{1\textrm{~MByte}}{1\textrm{~Gbps}} = 8\textrm{~ms}$.
With seamless flow switching, \sys completes at $\sim$$42$ ms due to protocol
overhead ($\sim$$3\%$ bandwidth loss due to TCP/IP header) and first-flow initialization time (two-RTT latency loss; one RTT latency for the sender to receive the SYN-ACK,
and an additional RTT for the sender to receive the first DATA-ACK).
We observe that \sys can converge to equilibrium quickly at flow switching time,
resulting in a near perfect ($100\%$) bottleneck link utilization (Figure~\ref{fig:b_util}).
Although \matt{an alternative (naive)} approach to achieve such high link utilization is to let
every flow send with fastest rate, this causes the rapid growth of the queue and
potentially leads to congestive packet loss.
Unlike this approach, \sys exhibits a very small queue size\footnote{The non-integer values on the y axis comes from the small probing packets.}
and has no packet drops (Figure~\ref{fig:b_queue}).

\begin{figure}[t]
\centering
\subfloat[]{ \label{fig:b_throughput}   \includegraphics[width=3.2in]{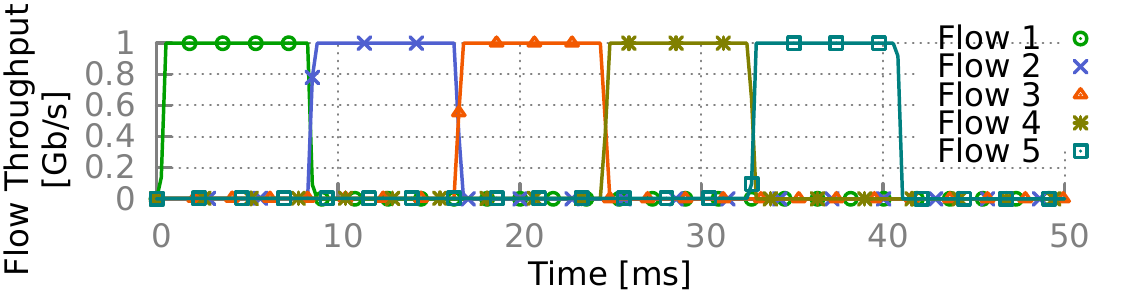}}\\\objtoobj
\subfloat[]{ \label{fig:b_util}   \includegraphics[width=3.2in]{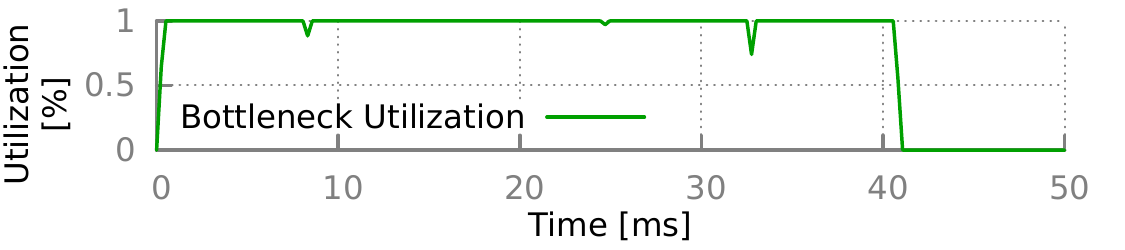}}\\\objtoobj
\subfloat[]{ \label{fig:b_queue}   \includegraphics[width=3.2in]{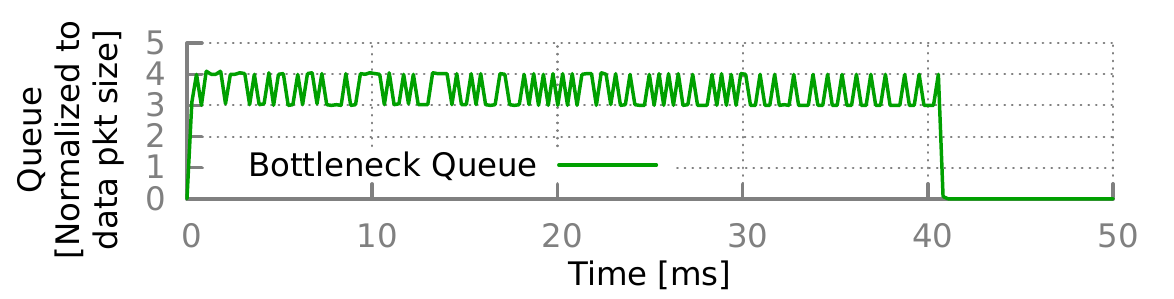}}\objtolbl
\caption{\figtitle{\sys provides seamless flow switching. It achieves high link utilization at flow switching time, maintains small queue, and converges
to the equilibrium quickly.}}
\label{fig:case1}
\end{figure}

\paragraphb{Scenario \#2 (Robustness to Bursty Traffic):} Figure~\ref{fig:case2} shows that \sys provides high robustness to bursty workloads.
We assume a long-lived flow that starts at time $0$, and $50$ short flows that all start at $10$ ms.
%\mattc{this seems to be the first time the word "background" is used in the text. what is a "background" flow? I think you need to add a brief definition. my confusion is I think the traditional meaning of "Background" in data centers is a flow with lower priority, but your scheme doesn't do anything with priorities, so I don't see why doing simulations with bakground flows would make any difference. why are you using background flows? Also, are the   short flows background flows too? it would be nice to reword the sentence to eliminate any possible vagueness.}
The short flow sizes are set to $20$ KByte with small random perturbation.
Figure~\ref{fig:m_throughput} shows that \sys adapts quickly to sudden bursts of flow arrivals.
Because the required delivery time of each short flow is very small ($\frac{20\textrm{~KByte}}{1\textrm{~Gbps}} \approx 153$ $\mu$s),
the system never reaches stable state during the preemption period (between $10$ and $19$ ms).
Figure~\ref{fig:m_util} shows \sys adapts quickly to the burst of flows while maintaining high utilization:
the average link utilization during the preemption period is $91.7\%$.
%Moreover, Figure~\ref{fig:m_util} demonstrates that \sys lets short flows send one after one, resulting in near-optimal average completion time \mattc{how does sending flows one after another result in near-optimal average completion time? Do you mean that \sys's fast switchover enables that? Or the ability to follow EDF scheduling enables that?}.
Figure~\ref{fig:m_queue} suggests that \sys does not compromise the queue length by having
only $5$ to $10$ packets in the queue, which is about an order of magnitude smaller than what today's data center switches can store.
By contrast, XCP in a similar environment results in a queue of $\sim$$60$ packets (Figure 11(b) in~\cite{xcp}).

\begin{figure}[t]
\centering
\subfloat[]{ \label{fig:m_throughput}   \includegraphics[width=3.2in]{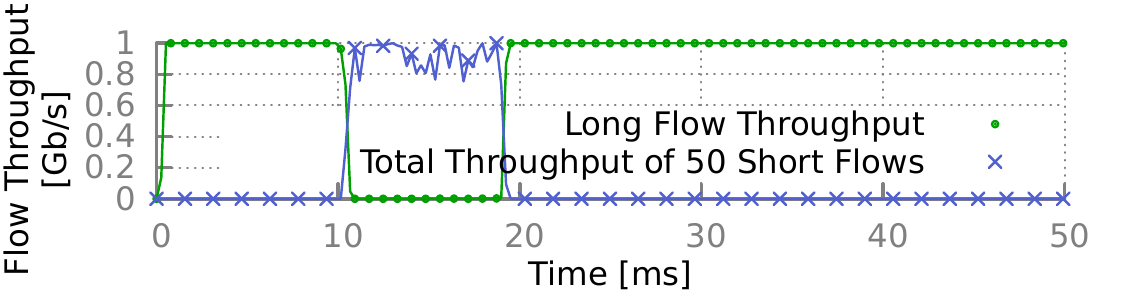}}\\\objtoobj
\subfloat[]{ \label{fig:m_util}   \includegraphics[width=3.2in]{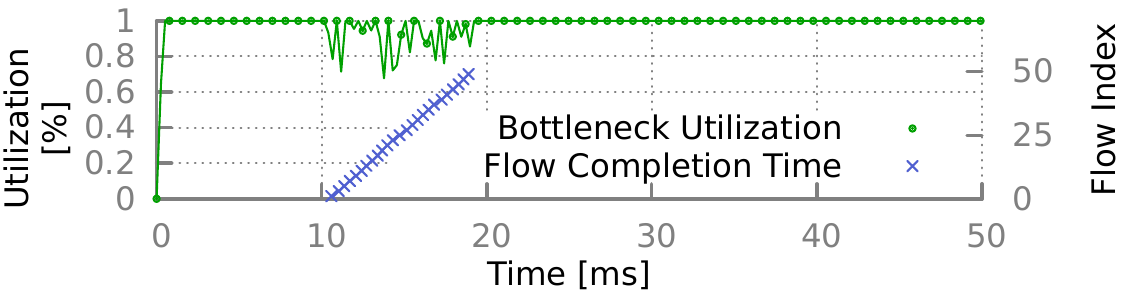}}\\\objtoobj
\subfloat[]{ \label{fig:m_queue}   \includegraphics[width=3.2in]{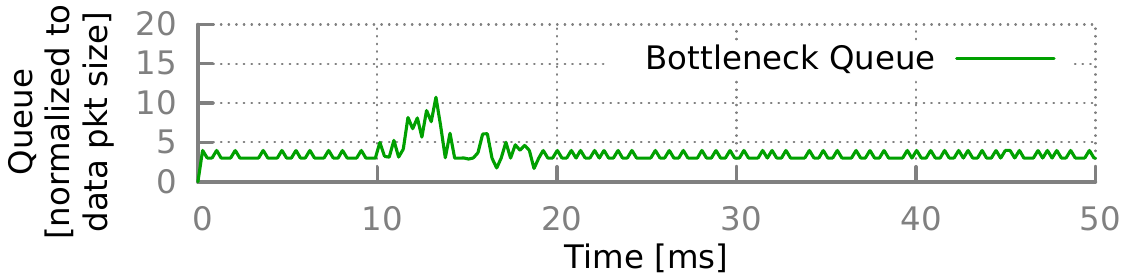}}\objtolbl
\caption{\figtitle{\sys exhibits high robustness to bursty workload.
We use a workload of $50$ concurrent short flows all start at time $1$ ms, and preempting a long-lived flow.}}
\label{fig:case2}
\end{figure}

\subsection{Impact of Network Scale}
\label{sec:perf:scale}
Today's data centers have many thousands of servers, and it remains unclear whether \sys will retain its successes at large scales.
Unfortunately, our packet-level simulator, which is optimized for high processing speeds, does not scale to 
large-scale data center topology within reasonable processing time. To study these protocols at large scales, we construct a 
\emph{flow-level simulator} for \sys, \bntl and RCP. In particular, we use an iterative approach 
to find the equilibrium flow sending rates with a time scale of $1$ ms. The flow-level simulator also considers 
protocol inefficiencies like flow initialization time and packet header overhead. Although the flow-level simulator does not deal with
packet-level dynamics such as timeouts or packet loss, Figure~\ref{fig:scale} shows that, by comparing with the 
results from packet-level simulation, the flow-level simulation does not compromise the accuracy significantly. 

We evaluate three scalable data center topologies: (1) Fat-tree~\cite{fattree}, a structured 2-stage Clos network; (2) BCube~\cite{bcube}, a server-centric 
modular network; (3) Jellyfish~\cite{jellyfish}, an unstructured high-bandwidth network using random regular graphs. 
Figure~\ref{fig:scale} demonstrates that \sys scales well to large scale, regardless of the topologies we tested.
Figure~\ref{fig:traffic:CDF} shows that about $40\%$ of flow completion times under \sys are reduced by at least $50\%$ compared to RCP. 
Only $5 - 15\%$ of the flows have a larger completion time, and no more than $0.9\%$ of the flows have $2\times$ completion time.

\begin{figure}[t]
\centering
\subfloat[]{ \label{fig:traffic}   \includegraphics[width=3.2in]{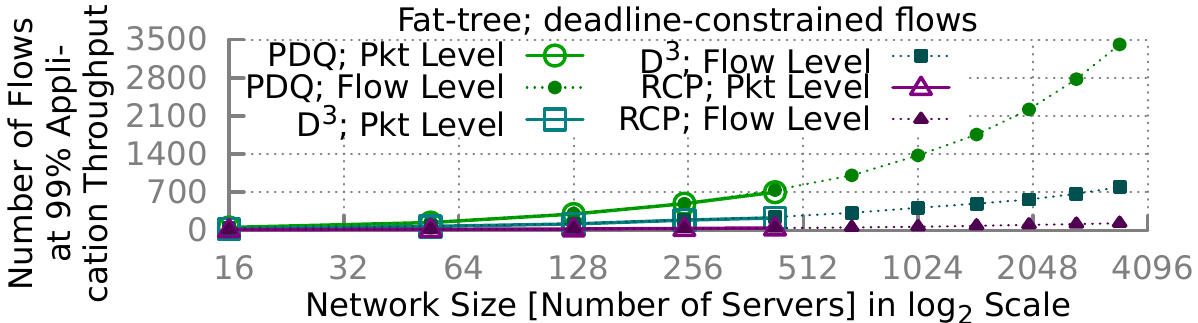}}\\\objtoobj
\hspace{7pt}\subfloat[]{ \label{fig:traffic:ND}   \includegraphics[width=3.2in]{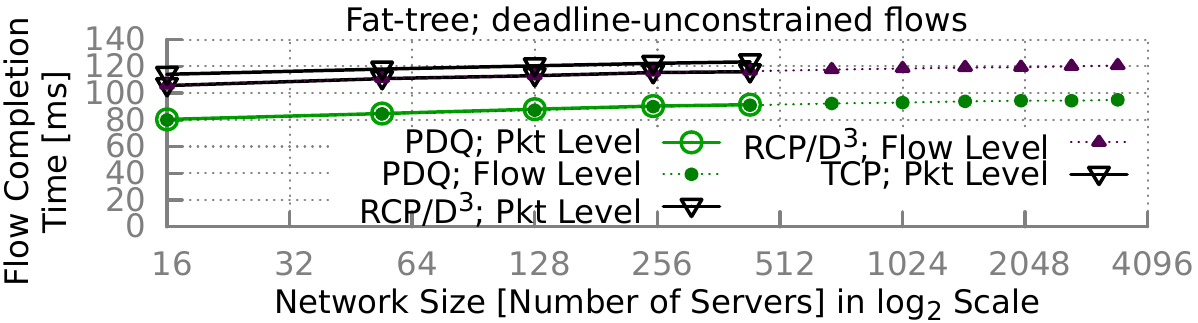}}\\\objtoobj
\subfloat[]{ \label{fig:traffic:NDBCube}   \includegraphics[height=1.2in]{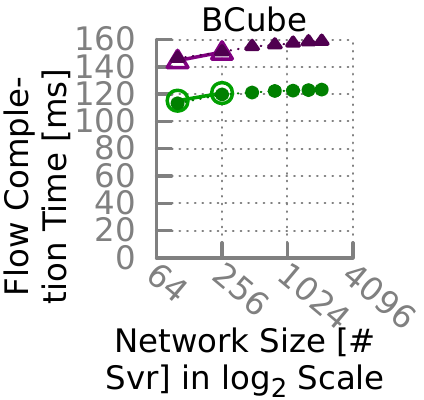}}\hspace{3pt}
\subfloat[]{ \label{fig:traffic:NDJellyfish}   \includegraphics[height=1.2in]{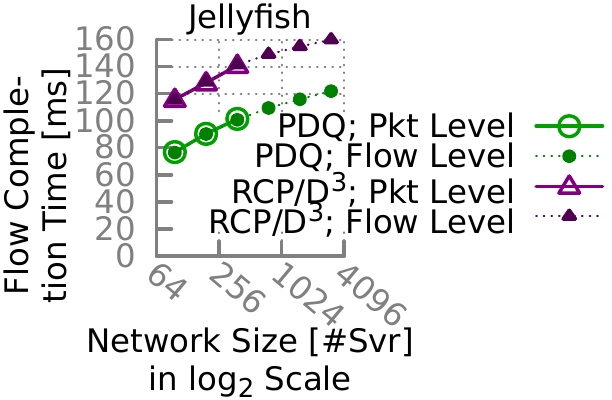}}\\\objtoobj
\subfloat[]{ \label{fig:traffic:CDF}   \includegraphics[width=3.2in]{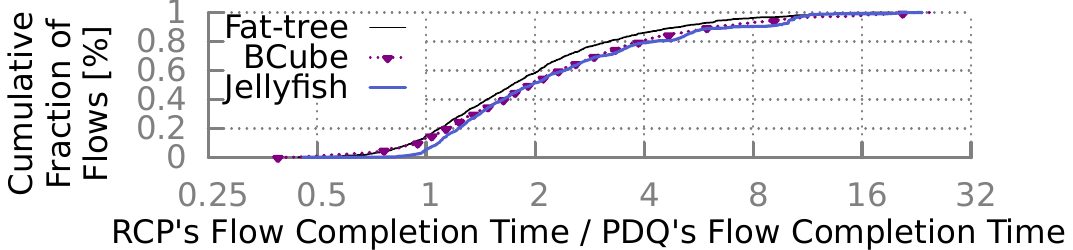}}\objtolbl
%<<<<<<< .mine
%\caption{\figtitle{\sys  \matt{performs well} across a variety of data center topologies. (a,b) Fat-tree; (c) BCube 
%=======
\caption{\figtitle{\sys performs \matt{well} across a variety of data center topologies. (a,b) Fat-tree; (c) BCube 
%>>>>>>> .r8705
with dual-port servers; (d) Jellyfish with $24$-port switches, using a $2$$:$$1$ ratio of network port count to server port count.
(e) For network flows, the ratio of the flow completion time under \sys to the flow completion time under RCP (flow-level simulation; \# servers is $\sim$$128$).
All experiments are carried out using random permutation traffic; top figure: deadline-constrained flows; bottom four 
figures: deadline-unconstrained flows with $10$ sending flows per server.
}}
\label{fig:scale}
\end{figure}

%% file: badcases.tex
\subsection{\sys Resilience}
\label{sec:perf:badcases}

\paragraphb{Resilience to Packet Loss:}
\matt{
Next, to evaluate \sys's performance in the presence of packet loss,}
%Next, while we know \sys remains correct in the presence of packet loss
%(\S\ref{sec:properties}), 
%While we know \sys has no deadlock and will converge to equilibrium  even when communication is lossy (\S\ref{sec:properties}),
%Although we already show that \sys has no deadlock and will converge to equilibrium even when communication is lossy (\S\ref{sec:properties}),
%it is interesting to see how \sys is resilient to unexpected packet loss. 
%To evaluate this, 
we randomly drop packets at the bottleneck 
link, in both \matt{the} forward (data) and reverse (acknowledgment) direction. 
Figure~\ref{fig:error} demonstrates that \sys is even more resilient than TCP
\matt{to packet loss}. 
When packet loss happens, the \sys rate controller detects anomalous high/low link load quickly 
and compensates for it with explicit rate control. Thus, packet loss does not significantly 
affect its performance. For a heavily lossy channel where the packet loss rate is $3\%$ in both directions (i.e., a round-trip packet loss rate 
of $1-(1-0.03)^2\approx5.9\%$), as shown in Figure~\ref{fig:error:ND}, the flow completion time of \sys has increased by $11.4\%$, while 
that of TCP has significantly increased by $44.7\%$.

\begin{figure}[t]
\centering
\subfloat[]{ \label{fig:error:D}   \includegraphics[width=1.6in]{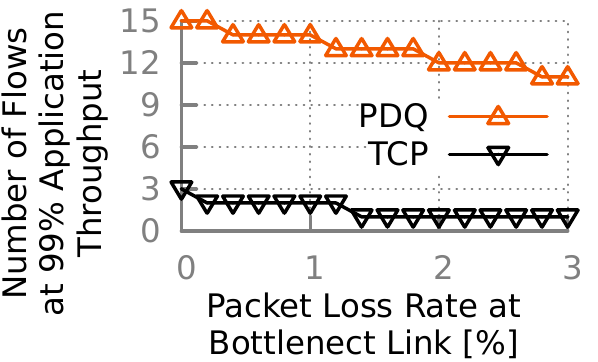}}
\subfloat[]{ \label{fig:error:ND}   \includegraphics[width=1.6in]{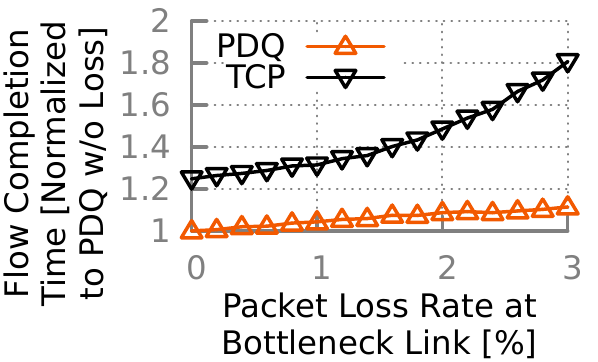}}\objtolbl
\caption{\figtitle{\sys is \matt{resilient to packet loss in both forward and reverse directions:}
%As the packet loss 
%rate increases, \sys degrades gracefully and slower than TCP does in both 
(a) deadline-constrained and (b) deadline-unconstrained cases. Query aggregation workload.}}\label{fig:error}
\end{figure}

\paragraphb{Resilience to Inaccurate Flow Information:}
%The \sys sender requires the knowledge of flow size at flow initiation time, and 
%we discuss why this information is feasible in \S\ref{sec:discussion}.
\cready{For many data center applications (e.g., web search, key-value stores, data processing), previous
studies have shown that flow size can be precisely known at flow initiation time.\footnote{See the discussion in \S 2.1 of \cite{bntl}.}
Even for applications without such knowledge, \sys is resilient to inaccurate flow size information. 
To demonstrate this, we consider the following two flow-size-unaware schemes. \emph{Random}: the sender randomly 
chooses a flow criticality at flow initialization time and uses it consistently. \emph{Flow Size Estimation}: the sender estimates the flow size 
based on the amount of data sent already, and a flow is more critical than another one if it has smaller estimated size. 
To avoid excessive switching among flows, the sender does not change the flow criticality for every packet it sends. Instead, the sender 
updates the flow criticality only for every $50$ KByte it sends. Figure~\ref{fig:flowEstimation} demonstrates two important results: (i) \sys does require 
reasonable estimate of flow size as random criticality can lead to large mean flow completion time in 
heavy-tailed flow size distribution. (ii) With a simple estimation scheme, \sys still compares 
favorably against RCP in both uniform and heavy-tailed flow size distributions.}

%This demonstrates that \sys gains from not only the specific scheduling disciplines used but also the ability of scheduling flows to send one after another.

\cut{We demonstrate that \sys is resilient to inaccurate flow size 
information by adding a random error. 
For a flow with actual size $s$, to introduce inaccuracy, we set the flow size information in the scheduling header to $s'$.
We say the flow size inaccuracy is $p$ if $s'$ is uniformly drawn from [$s/(1+p)$, $s\times(1+p)$].
We fix the number of flows at $10$. 
Figure~\ref{fig:inaccurateflowsize} demonstrates that, even when the flow size information are 
highly inaccurate, the flow completion time of \sys increased only by $\sim$$9\%$.
We also experiment with random flow criticality order, which provides almost 
the same performance as when the flow size inaccuracy is $500\%$, suggesting that 
rough estimation of flow size information causes only minor performance degradation.
This demonstrates that \sys gains from scheduling flows to send one after another, rather than the specific scheduling disciplines used.
}

%\begin{figure}[t]
%\centering
%\subfloat[]{ \label{fig:inaccurateflowsize:small}   \includegraphics[width=1.5in]{fig/InaccurateFlowSize.pdf}}\hspace{4pt}
%\subfloat[]{ \label{fig:inaccurateflowsize:large}   \includegraphics[width=1.5in]{fig/InaccurateFlowSize_Large.pdf}} \objtolbl
%\caption{\figtitle{\sys is \matt{resilient} to inaccurate flow information. Query aggregation workload with $10$ deadline-unconstrained flows.}}
%\label{fig:inaccurateflowsize}
%\end{figure}

\begin{figure}[t]
\centering
\includegraphics[width=3.25in]{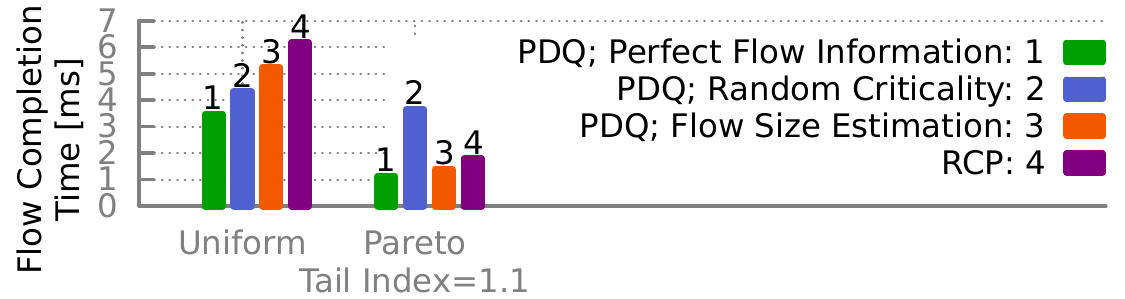} \objtolbl
\caption{\figtitle{\cready{\sys is resilient to inaccurate flow information. 
For PDQ without flow size information, the flow criticality is updated for every $50$ KByte it sends. 
Query aggregation workload, $10$ deadline-unconstrained flows with a mean size of $100$ KByte. Flow-level simulation.}}}
\label{fig:flowEstimation}
\end{figure}

%\vspace{-5pt}

%% file: multipath.tex
\section{Multipath \sys}
\label{sec:multipath}
Several recent works~\cite{mptcp,mptcp-nsdi} \matt{show} the benefits of multipath TCP, 
ranging from improved reliability to higher network utilization.
Motivated by \matt{this work}, we propose Multipath \sys (\mpsys), \matt{which enables a single flow to}
be striped across multiple network paths.
%In multihoming case, the \mpsys sender hashes subflows to different IP addresses it owns.

%\matt{one for each path} ---> this is incorrect. We use ECMP

When a flow \matt{arrives}, the \mpsys sender splits the flow into multiple subflows, and sends a SYN packet for each subflow.  
To minimize the flow completion time, the \mpsys sender periodically shifts the load from 
the paused subflows to the sending one with the minimal remaining load. 
%We slightly modify the early termination algorithm to take the parallel subflows into account.
To support \mpsys, the switch uses flow-level Equal-Cost MultiPath (ECMP) to assign subflows to paths.
The \sys switch requires no additional modification \cready{except ECMP.}
%The \sys switch requires no additional modification. In particular, the \mpsys switch 
%simply considers subflows as different and unrelated flows. 
%In addition, the switch uses flow-level Equal-Cost MultiPath (ECMP) \cycomment{to randomly assign subflows to paths.}
The \mpsys receiver maintains a single shared buffer for a multipath flow to resequence out-of-order packet 
arrivals, as done in Multipath TCP~\cite{mptcp-nsdi}.

%multi-homed 
%<<<<<<< .mine
%data center 
%topology that allows \mpsys to exploit 
%the path diversity between hosts. \matt{A} BCube($n$,$k$) has $n^{k+1}$ servers and $(k+1) \times n^k $ switches.
%We implement the address-based multipath routing algorithm 
%proposed in BCube~\cite{bcube}. For any pair of servers, the algorithm is guaranteed to 
%have $k+1$ parallel paths. 
%=======
We illustrate the performance gains of \mpsys using BCube~\cite{bcube}, a data center topology that allows \mpsys to exploit 
the path diversity between hosts. 
%In particular, a BCube($n$,$k$) has $n^{k+1}$ servers and $(k+1) \times n^k $ switches.
%$k+1$ parallel paths in a BCube($n$,$k$) network.
\cycomment{We implement BCube address-based routing to derive multiple parallel paths.}
%algorithm proposed in BCube~\cite{bcube}. For any pair of servers, the algorithm is guaranteed to 
%have $k+1$ parallel paths. 
%>>>>>>> .r8705
%We rely on flow-level ECMP to hash subflows to parallel paths randomly.
Using random permutation traffic, Figure~\ref{fig:mpscp:increasingworkload:ND} demonstrates 
the impact of the system load on flow completion time of \mpsys. 
\matt{Here,} we split a flow into $3$ \mpsys subflows. 
\matt{Under light loads, \mpsys can reduce flow completion time by a factor
of two. This happens because \mpsys exploits more links that are underutilized
or idle than single-path \sys.
As load increases, these advantages are reduced, since even single-path \sys
can saturate the bandwidth of nearly all links.
}
%When the system load is light, we observe that \mpsys can reduce the flow completion time by a factor of two. 
%Intuitively, \mpsys has \matt{a} smaller flow completion time because it exploits more links that are underutilized 
%or idle when applying single-path \sys. As system load increases, such advantages disappear since even single-path \sys 
%can saturate the bandwidth of almost every link.
However, as \matt{shown} in Figure~\ref{fig:mpscp:increasingworkload:ND}, \mpsys still 
retains its benefits because \mpsys allows a critical flow to have 
higher sending rate by utilizing multiple parallel paths. 
\matt{
Finally, we fix the workload at $100\%$ to stress the network (Figures~\ref{fig:mpscp:increasingsubflow:ND} and \ref{fig:mpscp:increasingsubflow}).
We observe that \mpsys needs about $4$ subflows to reach $97\%$ of its full
potential.
By allowing servers to use all four interfaces (whereas single-path \sys can use only one),
\mpsys provides a significant performance improvement.
}

%We \matt{next} fix the workload at $100\%$ to stress the network. Figure~\ref{fig:mpscp:increasingsubflow:ND} and \ref{fig:mpscp:increasingsubflow} show
%the performance achieved as a function of the number of subflows used. 
%We observe that \mpsys needs about $4$ subflows to reach $97\%$ of its fullest potential.
%By allowing a server to use all four interfaces (whereas single-path \sys can use only one), 
%\mpsys remarkably elevates the performance for both deadline-constrained and unconstrained flows.

\begin{figure}[tb]
\centering
\hspace{-10pt}\subfloat[]{ \label{fig:mpscp:increasingworkload:ND}   \includegraphics[width=3.0in]{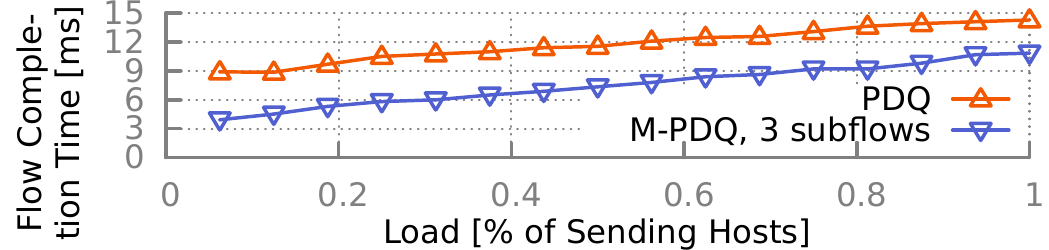}}\\\objtolbl
\subfloat[]{ \label{fig:mpscp:increasingsubflow:ND}   \includegraphics[width=1.5in]{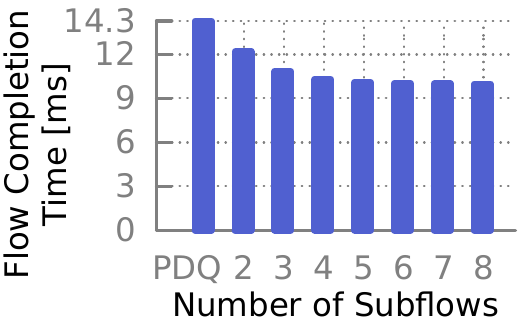}} \hspace{.1in}
\subfloat[]{ \label{fig:mpscp:increasingsubflow}   \includegraphics[width=1.5in]{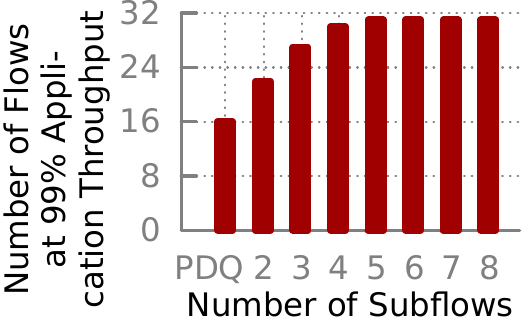}}\objtolbl
\caption{\figtitle{Multipath \sys achieves better performance. BCube(2,3) with 
random permutation traffic. (a, b) \matt{deadline-unconstrained,} (c) deadline-constrained flows.}}
\label{fig:multipath}
\end{figure}

%% file: discussion.tex
%\vspace{-12pt}
\section{Discussion}
\label{sec:discussion}

\cut{
\paragraphe{Flow Size Estimation.} For many data center applications (e.g., web search, key-value stores, data processing), previous 
studies have shown that flow size can be precisely known at flow initiation time.\footnote{See the discussion in \S 2.1 of \cite{bntl}.}
Even for applications without such knowledge, there are good schemes to accurately estimate flow sizes, for
example by matching based on packet header fields, by applying machine learning techniques
to classify flows~\cite{flowclassifier}, or by monitoring the server-side
buffer~\cite{mahout}.  The sender can also estimate the flow size based on the
amount of data sent already, as done in~\cite{devoflow}.  We also demonstrate
that \sys preserves nearly all its performance gains even given inaccurate
flow information (\S\ref{sec:perf:badcases}).
}

%For almost all the interactive applications (e.g., web search, key-value stores, 
%data processing), 
%the previous study showed that the flow size can be known in advance~\cite{bntl}. 
%Even for the applications without such knowledge, there are existing approaches for this already.
%Thus, the data center operator can provide a good estimate of the expected flow size. For example, one can use 
%either a simple matching based on the packet header fields or machine learning techniques to classify flows~\cite{flowclassifier}.
%Alternatively, by monitoring server-side buffer, Mahout~\cite{mahout} can help to determine a flow size much earlier than the flow 
%finishes its delivery. Furthermore, the sender can better estimate the flow size based on the amount of data it sent 
%already, as done in~\cite{devoflow} to detect elephant flows. 
%We also demonstrate that \sys preserves nearly all its performance gains even given inaccurate flow information.

\paragraphb{Fairness.} 
One could argue the performance gains of \sys over other 
protocols stem from the fact that \sys unfairly penalizes less critical flows.
%For example, it seems intuitive to assume that the huge \emph{average} flow 
%completion time improvements of SJF over fair sharing 
%comes from penalizing the \emph{large} flows.
Perhaps counter-intuitively, the performance gain of SJF over fair sharing does 
not usually come at the expense of long jobs.
An analysis~\cite{PShaslargeFCT} shows that at least $99\%$ of jobs 
have a smaller completion time under SJF than under fair 
sharing, and 
%the percentage can go to $100\%$ 
\matt{this percentage increases further} when the traffic load 
is less than half.\footnote{Assuming a M/G/1 queueing model 
with heavy-tailed flow distributions; see~\cite{PShaslargeFCT}.} 
Our results further demonstrate that, even in complex data center networks with thousands of concurrent flows and multiple bottlenecks, 
$85 - 95\%$ of \sys's flows have a smaller completion time than
%the \sys flows have a smaller completion as compared with 
RCP, and the worst \sys flow suffers an inflation factor of only
$2.57$ as compared with RCP
%worst \sys flow has an increased completion time only by a ratio of $2.57$ 
(Figure~\ref{fig:traffic:CDF}).
%every flow have more than 
%$2.57\times$ completion time under \sys than that under RCP .}
Moreover, unfairness might not be a primary concern in data center networks where the network is 
owned by a single entity that has full control of flow criticality. 
However, if desired, the operator can easily override the flow comparator to achieve a wide range of goals, including fairness. 
For example, to prevent starvation, the operator could gradually increase
the criticality of a flow based on its waiting time. \cycomment{Using a fat-tree topology with $256$ servers, Figure~\ref{fig:aging} demonstrates 
that this ``flow aging'' scheme is effective, reducing the worst flow completion time by $\sim$$48\%$, while the mean flow completion 
time increases only $1.7\%$.}
%Alternatively, the operator can override the comparator based on application-level information such as service level agreements.

\begin{figure}[t]
\centering
\subfloat[]{   \includegraphics[width=3.2in]{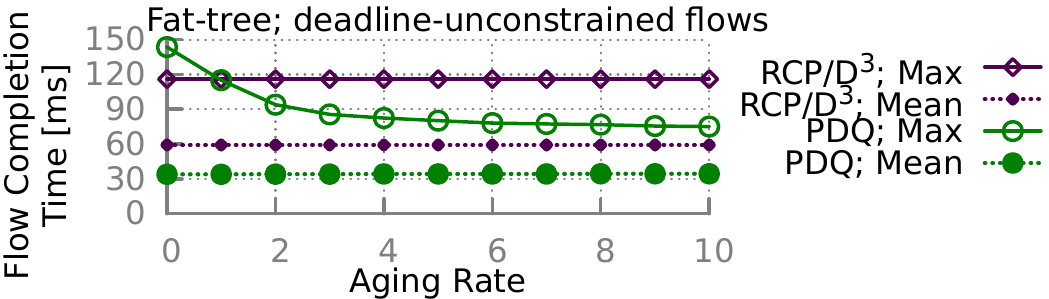}}\objtolbl
\caption{\figtitle{
\matt{Aging helps prevent less critical flows from starvation and shortens}
%Aging is effective to prevent less critical \sys flows from starvation and shorten 
their completion time. 
The \sys sender increases flow criticality by reducing \retimeh by a factor of $2^{\alpha t}$, where $\alpha$ is a parameter that controls the aging 
rate, and $t$ is the flow waiting time (in terms of $100$ ms). Flow-level simulation; $128$-server fat-tree topology; random permutation traffic.
}}\label{fig:aging}
\end{figure}

\paragraphb{When flow completion time is not the priority.} 
Flow completion time is not the best metric for some protocols.
For example, real-time audio and video may require the ability to 
{\em stream}, or provide a number of flows with a fixed fraction 
of network capacity. For these applications, protocols designed for streaming transport \cut{such as DCCP, RTP, and SCTP,}may be a better fit.
\cready{One can configure the rate controller~(\S\ref{sec:algorithm:switch:rc}) to slice the network into
PDQ-traffic and non-PDQ-traffic, and use some other transport protocol for non-PDQ-traffic.}
In addition, there are also applications where the receiver
may not be able to process incoming data at the full line rate.
In such cases, sending any rate faster than what receiver can process does not offer 
substantial benefits.
Assuming the receiver buffers are reasonably small, PDQ will back off and allocate remaining bandwidth to another flow.

\paragraphb{Does preemption in PDQ require rewriting applications?}
\cready{A preempted flow is paused (briefly), not terminated. From the
application's perspective, it is equivalent to TCP being slow
momentarily; the transport connection stays open. Applications do not
need to be rewritten since preemption is hidden in the transport layer.}

\paragraphb{Incentive to game the system.} Users are rational and may have an incentive to 
improve the completion time of their own flows by splitting each flow into small flows.
While a similar issue happens to \bntl, TCP and RCP\footnote{In TCP/RCP, users may achieve higher aggregated throughput by splitting
a flow into smaller flows; in \bntl, users may request a higher rate than the flow actually needs.},
users in \sys may have an even greater incentive, since \sys does preemption.
\cycomment{It seems plausible} to penalize users for having a large number of short flows by reducing their flows' criticality.
Developing a specific scheme remains as future work.

\paragraphb{Deployment.} On end hosts, one can implement \sys by inserting a shim layer between the IP and the transport layers. 
In particular, the sender maintains 
a set of \sys variables, intercepts all calls between IP 
and transport layer, attaches and strips off the \sys scheduling header\footnote{The $16$-byte scheduling header
consists of $4$ fields, each occupying $4$ bytes: \rateh, \pausebyh, \deadlineh, and \retimeh. 
The \sys receiver adds \pfs and \rtts to the header by reusing the fields used by \deadlineh and \retimeh. This is feasible 
because \deadlineh and \retimeh are used only in the forward path, while \pfs and \rtts are used 
only in the reverse path. Any reasonable hashing that maps switch ID to $4$-byte \pausebyh should 
provide negligible collision probability.}, and passes the packet segment to IP/transport layer 
accordingly. \cready{Additionally, the shim layer could provide an API that allows applications to specify the deadline and flow size, or it could avoid the API by
estimating flow sizes (\S\ref{sec:perf:badcases}).} \matt{The} \sys sender can easily override TCP's congestion window size to control the flow 
sending rate. \cready{We note that \sys requires only a few more variables per flow on end hosts.} 
On switches, similar to previous proposals such as \bntl, a vendor can implement \sys by making 
modifications to the switch's hardware and software. Per-packet operations like modifying header fields 
are already implemented on most vendors' hardware (e.g., ASICs), which can be directly used
by our design. The more complex operations like computing the aggregated 
flow rate and sorting/updating the flow list can be implemented in software. 
%\matt{If the vendor's hardware supports the ability to manipulate}
\cready{We note that \sys's per-packet running time is $O(\kappa)$ for the top $\kappa$ flows and $O(1)$ for the rest of the flows, where $\kappa$ is a small number of flows with the highest criticality and can be 
bounded as in \S\ref{sec:algorithm:switch:s}}. 
%has similar switch processing and state requirements to \bntl in most pratical cases. 
%We impose minimal switch state by maintaining the information of only a 
%small number of flows with the highest criticality.
%Moreover, the majority of \sys flows are paused and will require little computation at switches.
The majority of the sending flows' scheduling headers would remain 
unmodified\footnote{Until, of course, the flow is preempted or terminated.} by switches.

%Matt's text:
\cut{To improve deployability, we would like to avoid requiring vendors to make
hardware-level changes to their existing switches.  To support this, we
describe how \sys can be deployed through software-only modifications if the
vendor's equipment already supports OpenFlow.  OpenFlow is a technology that
exposes a set of APIs to allow fine-grained control over network forwarding
behavior.  To implement OpenFlow within a network device, the vendor extends
the hardware to provide customizable control on packet forwarding behavior, and
an interface to allow switch software to control this functionality.  The
vendor also implements the OpenFlow protocol, to enable a centralized
controller to control the switch's functionality.  \sys makes use of the
former, but not the latter.  In particular, \sys can be implemented as a set
of extensions to the software running at a switch, that exploits the local
OpenFlow hardware interfaces developed by the vendor.  \sys's hardware
operations (flow rate sampling and modifying header fields) can then be done by
directly manipulating OpenFlow rules (using matching and counters for flow rate
sampling, and actions for modifying header fields). As before, the remainder of
\sys's functionality can be implemented to the software running at the switch.}

\cut{Commercial switches can be thought of as specialized computers.  Like
traditional computers, they often run {\em software}: an operating system and
collection of protocol daemons that run atop a microprocessor.  The software of
a switch is used to implement routing protocol logic, configuration and vty
functionality, etc.  To support high-speed packet forwarding, switches often
also have customized {\em hardware}, often in the form of a customized ASIC.
This hardware support queuing, destination port lookup, filtering, matching,
and other operations that must be supported at high speed.  The switch's
software can control the hardware's operations through an interface, for
example, by writing into memory banks that are read directly by the switch's
hardware.}

\cut{In previous sections (Sections~\ref{xx-xx}, we described the switch-level
operational requirements of \sys.  \sys requires the ability to periodically
sample flow rates, to maintain a flow list data structure, to perform a
computation on the flow list to derive new flow rates, and to modify fields in
packet headers.  A vendor can implement \sys in a switch by making
modifications to the switch's hardware and software.  For example, operations
that must be done per-packet at high speeds (flow rate sampling and modifying
header fields)  can be implemented by modifying the switch's ASIC.  More
complex operations that do not need to be done in real time can be implemented
in software (computing the current set of flow rates for header fields, and
maintaining the flow list). \sys's hardware operations are limited to flow rate sampling and
modifying header fields, which can be done with straightforward modifications
to device hardware.}

%% file: related.tex
\vspace{-5pt}
\section{Related Work}
\label{sec:relatedwork}

%<<<<<<< .mine
%\paragraphb{Fair Sharing:} 
%\matt{TCP provides fair sharing across flows~\cite{aimd}.
%XCP~\cite{xcp} uses explicit congestion 
%feedback to ensure high stability and efficiency in high bandwidth-delay product networks.
%}
%% TCP adopts additive increase and multiplicative 
%% decrease (AIMD) to converge to fair share allocation of network bandwidth among flows~\cite{aimd}.
%=======

\cut{\paragraphb{Fair Sharing:} TCP adopts additive increase and multiplicative 
decrease (AIMD) to converge to fair share allocation of network bandwidth among flows~\cite{aimd}.
There is a rich literature on improving TCP performance. 
% XCP~\cite{xcp} uses control theory and explicit congestion 
% feedback to ensure high stability and efficiency regardless of bandwidth-delay product. 
By explicitly assigning a single fair-share rate at switches, RCP~\cite{rcp}
\matt{reduces} convergence time to reach fair-share rate than TCP\cut{ and XCP}.
\matt{DCTCP~\cite{dctcp} and ICTCP~\cite{ictcp} improve}
TCP burst tolerance in data center networks.
DCTCP modifies ECN (Explicit Congestion Notification) to alleviate 
TCP impairments in data center networks (e.g., TCP Incast), while
ICTCP targets at almost the same problem with an alternative approach: 
\matt{a} receiver-based congestion control that requires no modification on routers.
However, all of these protocols emulate fair sharing, which leads to suboptimal flow completion time.
}

%\paragraphb{\bntl:} 
%Wilson et al. propose \bntl, a deadline-aware congestion control protocol for data center networks~\cite{bntl}.
%\bntl explicitly controls the flow sending rate, and its objective is to minimize the number of flows missing their deadlines.
%While \bntl also employs explicit rate control like \sys, it neither resequences flow 
%transmission order nor preempts flows, resulting in a substantially different flow schedule which serves flows according to the order of their arrival.
%Unfortunately, this allows flows with large deadlines to hog the bottleneck bandwidth, blocking short flows that arrived later.

\paragraphb{\bntl:} 
%Wilson et al. propose \bntl, a deadline-aware congestion control protocol for data center networks~\cite{bntl}.
%\bntl explicitly controls the flow sending rate, and its objective is to minimize the number of flows missing their deadlines.
While \bntl~\cite{bntl} is a deadline-aware protocol that also employs explicit rate control like \sys, it neither resequences flow 
transmission order nor preempts flows, resulting in a substantially different flow schedule which serves flows according to the order of their arrival.
Unfortunately, this allows flows with large deadlines to hog the bottleneck bandwidth, blocking short flows that arrived later.

\paragraphb{Fair Sharing:} TCP, RCP~\cite{rcp} and DCTCP~\cite{dctcp}  all emulate fair sharing, which leads to suboptimal flow completion time. 

\paragraphb{TCP/RCP with Priority Queueing:} One could use priority queuing at switches and assigning different priority levels
to flows based on their deadlines. Previous studies~\cite{bntl} showed that, using 
two-level priorities, TCP/RCP with priority queueing suffers from losses and falls behind \bntl, and
increasing the priority classes to four does not significantly improve performance.
This is because flows can have very different deadlines and require a large number of priority classes, while
switches nowadays provide only a small number of classes, mostly no more than ten.

%One might ask how \emph{fair sharing with priority queueing} compares with \sys. In particular, one can implement 
%priority queueing at switches and assign different priority levels to flows based on its deadline/size. 
%Previous study~\cite{bntl} showed that, using two-level priorities, TCP/RCP with priority queueing suffers from losses and falls behind \bntl, and 
%increasing the priority classes to four does not improve performance much.
%This is because it requires a large number of priority classes to emulate \sys, while 
%switches nowadays provide only a small number of classes, mostly no more than ten.

\paragraphb{ATM:} \cready{One could use ATM to achieve QoS priority control. However, ATM's CLP classifies traffic into only two priority levels, while PDQ gives each flow a unique priority. Moreover, ATM is unable to preempt flows (i.e., new flows cannot affect existing ones). }

\paragraphb{DeTail:} In a \matt{recent} (Oct 2011) technical report, Zats et al. propose DeTail~\cite{detail}, an in-network multipath-aware congestion management mechanism that reduces
the flow completion time ``tail'' in datacenter networks. 
However, it targets neither mean flow completion time nor the number of deadline-missing flows. 
Unlike DeTail which removes the tail, \sys can save $\sim$$30\%$ flow completion time on average (compared with TCP and RCP), reducing the completion 
time of almost every flow (e.g., $85\% - 95\%$ of the flows, Figure~\ref{fig:traffic:CDF}). 
We have not attempted a direct comparison due to the very different focus and the recency of this work.

%% file: conclusion.tex
\section{Conclusion}
\label{sec:conclusion}

% matt: this sounds too much like D3's story...
%We raised doubts about whether fair sharing is the best principle for congestion control in data center networks.
%Theory suggests that fair sharing could be far from 
%optimal for satisfying flow requirements, and this performance gap 
%motivated us to step back and analyze the feasibility of applying 
%scheduling disciplines, many of which are known to have desirable properties.

We proposed \sys, a flow scheduling protocol designed to complete flows quickly
and meet flow deadlines. \sys provides a distributed algorithm to approximate a range of scheduling
disciplines based on relative priority of flows, minimizing mean flow completion time and the number of deadline-missing flows.
We perform extensive packet-level and flow-level simulation of \sys and several
related works, leveraging real datacenter workloads and a variety of traffic
patterns, network topologies, and network sizes.  We find that \sys provides
significant advantages over existing schemes.  In particular, \sys can reduce
by $\sim$$30\%$ the average flow completion time as compared with TCP, RCP and
D$^3$; and can support $3\times$ as many concurrent senders as D$^3$ while
meeting flow deadlines.  We also design a multipath variant of \sys by
splitting a single flow into multiple subflows, and demonstrate that \mpsys
achieves further performance and reliability gains under a variety of settings.

%% file: appendix.tex
%\vspace{-9pt}
%\prooffontsize
\section*{Appendix A. Deadlock-freedom}
Deadlock is a situation where two or more competing flows are paused and are each waiting for the other to
finish, and therefore neither ever does. We verify \sys has no deadlock by showing that \emph{hold and wait}, a necessary condition of
deadlock, is false. Hold and wait is a situation that a flow is accepted by some intermediate switches, while paused by other 
switches along the routing path.
In \sys, a flow is accepted \emph{only} after \emph{every} switch along the path accepts the flow.
Moreover, if a \sys flow is paused, the switch that pauses this
flow will update the pauseby field in the scheduling header (\pausebyh) to its ID.
Hence, after some time goes by, this information will reach all the other switches along the path, as
even the paused flows would send probes periodically.
%Even if packet loss happens, the same switch will update the blockby field on the next packet 
%sent by the sender of the paused flow.
Whenever a switch notices that a flow is paused by another switch, it will not consider accepting this flow.
Thus, a paused flow will not be accepted by \emph{any} switch along the path.

%\vspace{-9pt}
\section*{Appendix B. Bounding the Convergence Time}
\label{sec:appendix}

\paragraphb{Assumptions:} Without loss of generality, we assume there is no
packet loss. Similarly, we assume flows will not be paused due to the use of flow dampening.
Because \sys flows periodically send probes, the properties we discuss in this section will hold with additional latency when
the above assumptions are violated. For simplicity, we also assume the link rate $C$ is equal to
the maximal sending rate \maxrates (i.e., \schrates $=$ $0$ or $C$). Thus, each link accepts only one flow at a time.

\paragraphb{Definitions:}
%Recall that \sys switches share a common flow comparator algorithm to determine whether a flow is more critical than another (\S\ref{sec:algorithm:switch}).
We say a flow is \emph{competing} with another flow if and only if they share at least one common link.
Moreover, we say a flow $F_1$ is a \emph{precedential} flow of flow $F_2$ if and only if they are
competing with each other and flow $F_1$ is more critical than flow $F_2$.
We say a flow $F$ is a \emph{driver} if and only if (i) flow $F$ is more critical than any other competing flow, or (ii) all the
competing flows of flow $F$ that are more critical than flow $F$ are non-drivers.

\begin{lemma}
\label{theorem:accept}
\label{theorem:pause}
\textnormal{
When all the precedential flows of a flow $F$ are paused (or it has none) and workload is stable (no new flows arrive and no sending flow finishes), flow $F$ will 
be accepted in at most one RTT. If any precedential flow of a flow $F$ is accepted, flow $F$ will be paused in at most one RTT.
%If all the precedential flows of a flow $F$ are sending at their scheduling rate 
%(or it has none), then flow $F$ will send at the scheduling rate in at most two RTTs.
}
\end{lemma}
%\begin{proof}
%In this case, flow $F$ is more critical than any non-paused competing flow. 
%Although \sys usually takes only one RTT for switches to accept flow $F$, in the worst case it could take up to two RTTs:
%In the first RTT, flow $F$ might be temporarily blocked by an intermediate switch due to flow dampening. 
%This can happen if another flow is accepted by an intermediate switch right before the first packet of flow $F$ arrives at the switch.
%However, after the first RTT, all the intermediate switches will store flow $F$'s information.
%Therefore, flow F will be accepted in the second RTT. Flow dampening cannot happen in
%the second RTT because switches will not accept other flows after it has stored flow $F$'s information.
%\end{proof}

%\begin{theorem}
%\textnormal{
%}\end{theorem}
\begin{proof}
%\prooffontsize
%\vspace{-3pt}
This property follows directly from the \sys flow controller algorithm.
\end{proof}

\begin{lemma}
%\vspace{-4pt}
%\label{theorem:converge}
\textnormal{
\sys will converge to the equilibrium in $P_{\max}+1$ RTTs for stable
workloads, where $P_{\max}$ is the maximal number of precedential flows of any flow.
Given a collection of active flows, the equilibrium is defined as a state where all the drivers are accepted while the remaining flows are paused.
}
\end{lemma}
\begin{proof}
%\vspace{-3pt}
We show that when the workload is stable (no new flows arrive and no sending flow finishes), a flow will
be accepted if it is a driver and will otherwise be paused in $P + 1$ RTTs, where $P$ is the number of its precedential flows.
We prove this will hold for any flow $F$ that is the $m$-th critical flow in the network by induction on $m$.
%, the number of its competing flows that are more critical than flow $F$.
When $m=1$, flow $F$ is a driver by definition. %and is the most critical flow compared with its competing flows. 
Thus, it will be accepted in one RTT according to Lemma~\ref{theorem:accept}.
When $m=n+1$, there exist $n$ flows $F_1 \cdots F_n$ that are more critical than flow $F$.
Without loss of generality, out of these $n$ flows, we assume there are $n' \le n$ precedential flows (as they are competing with flow $F$).
Suppose that flow $F$ is a driver. Then, all these $n'$ flows are non-drivers by definition.
By the induction hypothesis, these $n'$ competing flows will all be paused in $P'+1$ RTTs, where $P'$ is the maximal possible number of precedential flows
of these $n'$ flows. As each of these $n'$ flows will have at most $n-1$ precedential flows, we have $P' \le n-1$.
After these flows are paused, it takes at most an RTT for switches to accept flow $F$ according to
Lemma~\ref{theorem:accept}, and therefore the flow $F$ will be
accepted in $(P'+1)+1 \le n+1$ RTTs.
Suppose now that flow $F$ is not a driver. According to the definition of driver, among those $n'$ competing flows, there exists $>$$0$ drivers.
Similarly, among these drivers, the maximal number of precedential flows of each driver is at most $n-1$.
By the induction hypothesis, these drivers will be accepted in at most $n$ RTTs, and
after this, flow $F$ will be paused in one RTT according to Lemma~\ref{theorem:pause}.
%When a driver arrives, we usually need one RTT latency before the sender can start sending. 
%After switches accept the driver, we need at most one RTT for the other non-driver flows 
%to stop sending after receiving the updated feedback.
%Thus, \sys will converge in at most $3$ RTTs when there is no packet loss. 
\end{proof}